\documentclass[11pt]{article}
\usepackage{titling, setspace}
\usepackage{geometry}
\usepackage[list=true, font=large, labelfont=bf, 
labelformat=brace, position=top]{subcaption} 
\usepackage{caption}
\usepackage{amsmath, amsthm, amssymb, amstext, amsfonts}
\usepackage{nicefrac}
\usepackage{enumerate}
\usepackage{graphicx}
\usepackage{color}
\usepackage{hyperref}

\newtheorem*{remark}{Remark}

\DeclareMathOperator*{\argmin}{argmin}
\DeclareMathOperator*{\argmax}{argmax}

\renewcommand{\phi}{\varphi}
\renewcommand{\epsilon}{\varepsilon}
\newcommand{\R}{\mathbb{R}}
\newcommand{\C}{\mathbb{C}}
\newcommand{\N}{\mathbb{N}}

\newcommand{\limit}[2]{\lim \limits_{#1 \to #2}}
\newcommand{\limn}{\limit{n}{\infty}}
\newcommand{\gegen}[2]{\xrightarrow{#1 \to #2}}
\newcommand{\drho}{\tfrac{\mathrm{d}}{\mathrm{d} \rho}}
\newcommand{\drs}{\tfrac{\mathrm{d}}{\mathrm{d} r_{s}}}

\newcommand{\E}{\mathcal{E}}
\newcommand{\A}{\mathcal{A}}
\newcommand{\M}{\mathcal{M}}
\newcommand{\intr}{\int_{\R^{3}}}

\newcommand{\di}{\, \mathrm{d}}
\newcommand{\e}{\mathrm{e}}
\newcommand{\eh}{\hat{\e}}

\numberwithin{equation}{section}
\newtheorem{theorem}{Theorem}
\newtheorem{lem}{Lemma}

\title{Existence and nonexistence of HOMO-LUMO excitations \\ in Kohn-Sham density functional theory}
\author{Gero Friesecke and Benedikt Graswald \\[1mm] 
\small Department of Mathematics, Technische Universit\"at M\"unchen \\[-1mm]
\small {\tt gf@ma.tum.de, graswabe@ma.tum.de}}

\begin{document}

\maketitle

\vspace*{-3mm}

\begin{abstract}
In numerical computations of response properties of electronic systems, the standard model is Kohn-Sham density functional theory (KS-DFT). Here we investigate the mathematical status of the simplest class of excitations in KS-DFT, HOMO-LUMO excitations. We show using concentration-compactness arguments that such excitations, i.e. excited states of the Kohn-Sham Hamiltonian, exist for $Z>N$, where $Z$ is the total nuclear charge and $N$ is the number of electrons. The result applies under realistic assumptions on the exchange-correlation functional, which we verify explicitly for the widely used PZ81 and PW92 functionals. By contrast, and somewhat surprisingly, we find using a method of Glaser, Martin, Grosse, and Thirring \cite{glaser1976} that in case of the hydrogen and helium atoms, excited states do not exist in the neutral case $Z=N$ when the self-consistent KS ground state density is replaced by a realistic but easier to analyze approximation (in case of hydrogen, the true Schr\"{o}dinger ground state density). Implications for interpreting minus the HOMO eigenvalue as an approximation to the ionization potential are indicated.
\end{abstract}
\vspace*{3mm}

\begin{small}
\begin{spacing}{0.01}
\tableofcontents
\end{spacing}

\end{small}

\noindent 

\section{Introduction}

Electronic excitations play an important role in the description of molecular properties such as absorption spectra, photoexcitation, state-to-state transition probabilities, reactivity, charge transfer processes, and reaction kinetics \cite{cramer2002essentials, Dai_Ho1995,Ho1996}.
In numerical computations of these response properties, the standard model is Kohn-Sham density functional theory (KS-DFT), because of its good compromise between accuracy and feasibility for large systems (see \cite{ParrYang1994} for a textbook account and \cite{Becke2014} for a recent review). It is then of interest to investigate the mathematical status of excitations in KS-DFT.

In this paper we analyze the simplest such excitations, HOMO-LUMO transitions, in the setting of the local density approximation (LDA). For a systematic comparison of HOMO-LUMO excitations with experimental data see e.g. \cite{Baerends2013, Zhang2007}. Even in this case we are not aware of previous rigorous results. Our findings are the following (see Figure \ref{fig:spectra}).

\begin{figure}[h]
    \centering
    \includegraphics[width = 0.65 \textwidth]{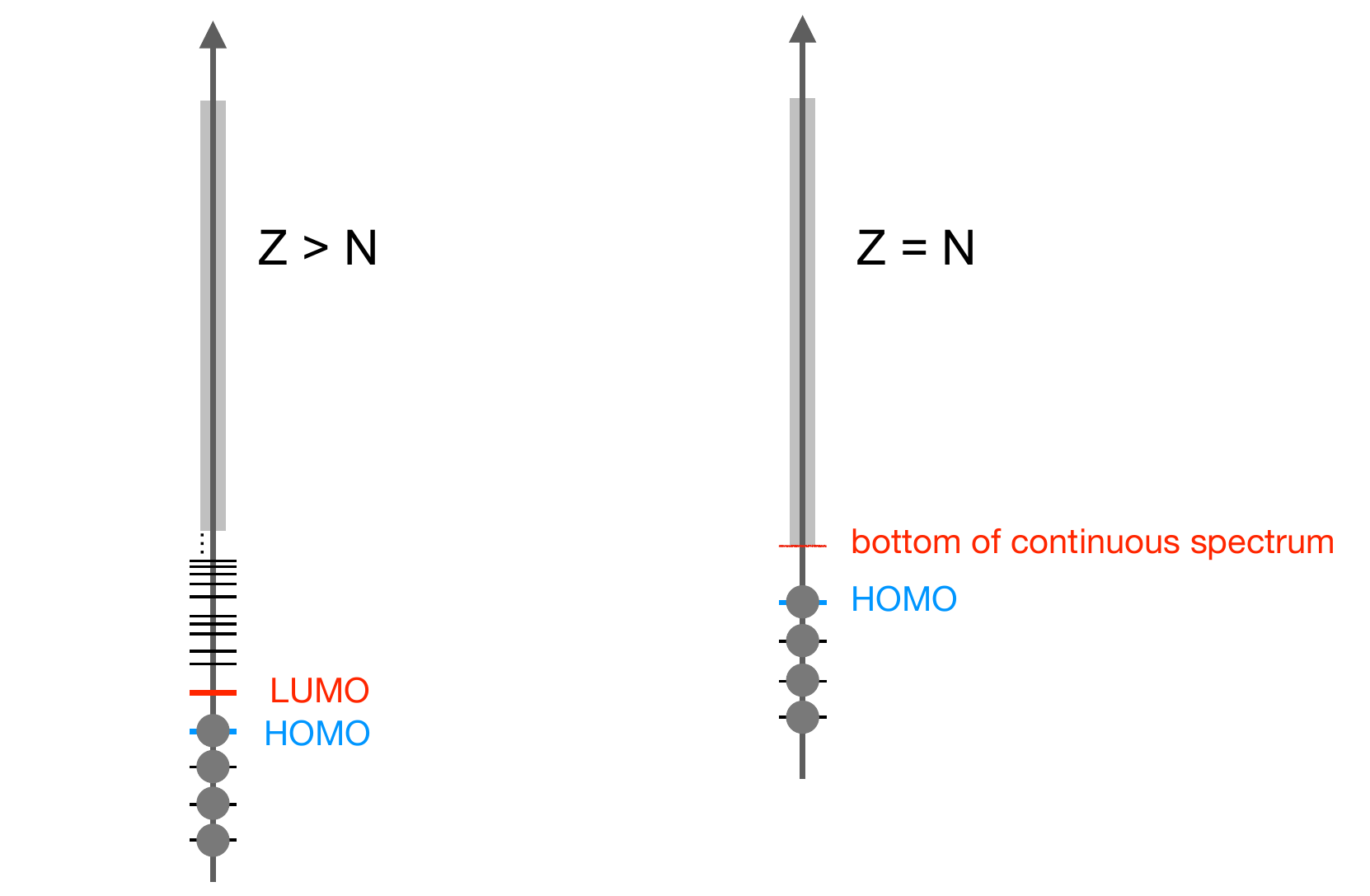}
    \caption{Schematic picture of the spectrum of the KS Hamiltonian. Positively charged systems (left, $Z>N$) have infinitely many excited states above the HOMO and below the continuous spectrum (see Theorem \ref{thm:neutral_case}). For neutral systems (right, $Z = N$), it can happen that there are no excited states, that is, the highest bound state eigenvalue is the HOMO (see Theorem \ref{thm:nonexistence}).}
    \label{fig:spectra}
\end{figure}

For positively charged systems (i.e., total nuclear charge $Z$ greater than the number $N$ of electrons) such excitations -- mathematically, excited states of the KS Hamiltonian -- are rigorously proven to exist, under realistic assumptions on the exchange-correlation functional which we verify explicitly for the widely used PZ81 and PW92 functionals. See Theorem \ref{thm:lumo} in Section 3.
As a corollary we also establish existence of {\it optimal} excitations with respect to suitable  control goals recently introduced in \cite{Friesecke-kniely}, without requiring the simplifying assumption in \cite{Friesecke-kniely} of bounded domains. See Section 4. 

By contrast, the neutral case $Z=N$ holds a surprise. In the case of the hydrogen and helium atoms, we prove that excited states do not exist when the self-consistent KS ground state density is replaced by a realistic but easier to analyze closed-form approximation (in case of hydrogen, the true Schr\"odinger ground state density).
See Theorem \ref{thm:nonexistence} in Section 5.

Mathematically, the existence result relies on concentration-compactness arguments, and should not come as a surprise to experts. The nonexistence result uses a not widely known method by Glaser, Martin, Grosse, and Thirring (GMGT) \cite{glaser1976}. The latter method could, in principle, also be applied to numerical KS ground state densities; we expect that for {\it some} atoms and molecules, including hydrogen and helium, the GMGT nonexistence criterion (that a certain integral associated with the effective KS potential lies below a threshold value) would be satisfied.

Physically, these results indicate a significant artefact of KS-DFT. In the full $N$-electron Schr\"odinger equation, neutral systems (and even systems with $Z>N-1$) are known to possess infinitely many excited states below the bottom of the continuous spectrum. 
This is a celebrated result by Zhislin \cite{zhislin1960discussion}; 
for a modern variational proof see \cite{Friesecke2003}. 
The analogous result also holds in Hartree-Fock theory: 
for $Z>N$ the Fock operator associated with the Hartree-Fock ground state density possess infinitely many bound states below the continuous spectrum  \cite[Lemma II.3]{Lions-Hartree}, the latter being the interval $[0, \infty)$.
(It is also known \cite{LewinHF} that the Hartree-Fock enegery functional possess infinitely many critical points below 0.)
Our results suggest that in KS-DFT, the threshold for existence of infinitely many excited states is shifted from $Z>N-1$ to $Z>N$. This is a previously unnoticed but important qualitative consequence of the (well known) incomplete cancellation of the self-interaction energy in KS-DFT.

It is interesting to interpret the nonexistence of excitations from the point of view of numerical computations in finite basis sets, or mathematical analysis (as in \cite{Friesecke-kniely}) in bounded domains. Consider a neutral system for which (exact) excitations do not exist. In a finite basis set, or a bounded domain, the spectrum of the KS Hamiltonian is purely discrete and therefore excited states exist. In the limit as the basis set approaches completeness, or the domain approaches the whole of $\R^3$, 

(i) the LUMO energy $\varepsilon_L$ (i.e., the lowest unoccupied eigenvalue of the KS Hamiltonian) will remain well-defined, and approaches the bottom of the continuous spectrum (which equals $0$, see Theorem \ref{thm:neutral_case} and Theorem \ref{thm:nonexistence})

(ii) the LUMO (i.e., the lowest unoccupied eigenstate) will become more and more delocalized, failing to converge to a bound state.

Thus in contrast to common (explicit or implicit) belief, restriction to finite basis sets or bounded domains may be not just a negligible technicality, but significantly alters the physical nature of LUMO excitations, from stable bound state (i.e., invariant under the dynamics of the KS ground state Hamiltonian) to a delocalized, dispersing state associated with the continuous spectrum. 

(ii) makes it very tempting to physically interpret the HOMO-LUMO excitation in the nonexistence case as an (approximation to an) ionization process. This interpretation together with (i) yields {\it ionization potential} $\approx \varepsilon_{L}-\varepsilon_{H}=0-\varepsilon_{H}$ (where $\varepsilon_H$ is the HOMO energy, i.e. the highest occupied eigenvalue of the KS Hamiltonian), 
lending new theoretical support to 
the famous semi-empirical formula 
$$
              -  \varepsilon_H  \; \approx \;  \mbox{\it ionization potential}     
$$
which often agrees quite well with experimental data \cite{Baerends2013, Zhang2007}.

\section{Mathematical setting \label{sec:mathematical_setting}}
We start by recalling well-known mathematical facts about Kohn-Sham density functional theory (KS-DFT) \cite{kohnsham65, hohenbergkohn64, ParrYang1994}.
Readers familiar with these facts might want to skip this part.
After that we give a variational definition of HOMO-LUMO excitations as introduced recently in \cite{Friesecke-kniely}, which works irrespective of degeneracies and is convenient for the mathematical analysis of excitations. 

\subsection{Kohn-Sham equations}

We consider a system of $N$ non-relativistic electrons in $\R^3$ in the electrostatic potential generated by $M$ nuclei of charges $Z_1, \ldots, Z_M$ located at positions $R_1, \ldots, R_M \in \R^3$,
\begin{equation} \label{eq:coloumb_special_case}
    v_{ext}(x) = - \sum_{\alpha = 1}^{M} Z_\alpha \frac{1}{|x - R_\alpha |}.
\end{equation}
In fact, for our analysis it is not essential that the nuclei are point particles.
It suffices to assume the nuclear charge distribution is given by a nonnegative Radon measure $\mu $ with total mass $Z>0$ supported on a compact set $\Omega_{nuc} \subseteq \R^3$, i.e. we consider any $\mu$ belonging to 
\begin{equation} \label{def:A_nuc}
 \mathcal{A}_{nuc} := \{ \mu \in \mathcal{M}(\Omega_{nuc}) : ~\mu \geq 0, ~ \int_{\Omega_{nuc}} \di \mu = Z \},  
\end{equation}
where $\mathcal{M}(\Omega_{nuc})$ denotes the space of signed Radon measures on $\Omega_{nuc}$, and 
\begin{equation}\label{eq:coloumb_potential}
    v_{ext}(x) := - \int_{\Omega_{nuc}} \frac{1}{|x-y|} \di \mu(y).
\end{equation}

For simplicity we look at a spin-unpolarized system, so the number $N$ of electrons is even, i.e. $N =2n$ for some $n \in \N$.
In this case Kohn-Sham DFT describes the electrons by $n$ orbitals $\phi_1, \ldots, \phi_n: \R^3 \to \C$, each occupied by two electrons of opposite spin.
They are 
$L^2$-orthonormal, i.e. 
\begin{equation} \label{eq:ortho}
   \langle \phi_i, \phi_j \rangle_{L^2} = \intr \phi_i(x) \phi_j(x) \di x = \delta_{ij} \qquad \forall i,j \in \{1,\ldots,n\},
\end{equation}
 and we denote  $\Phi := (\phi_1, 
\ldots, \phi_n)$.
Note that in the following $\langle \cdot, \cdot \rangle$ will always denote the $L^2$ inner product.
Then the corresponding Kohn-Sham energy functional is given by

\begin{equation}\label{eq:ks-energy-functional}
\hspace*{-5mm}
\E_\mu[\Phi] 
= 
\underbrace{\sum_{k=1}^{n} 2 \intr \frac{1}{2} | \nabla \phi_k|^2(x) \di x }_{=:T[\Phi]}
+ \underbrace{ \intr v_{ext}(x) \rho(x) \di x }_{=:V[\rho]}
+ \underbrace{\frac{1}{2} \intr \intr \frac{\rho(x) \rho(y)}{|x-y|}\di x \di y }_{=: J[\rho]}
+ \underbrace{ \intr e_{xc} (\rho(x)) \di x}_{=:E_{xc}[\rho]},
\end{equation}
where $e_{xc}$ gives the exchange-correlation energy per unit volume and $\rho$ is the total electron density, that is,

\begin{equation}
   \rho(x) := 2 \sum_{k=1}^n |\phi_k(x)|^2.
\end{equation}
The KS energy hence consists of the following terms:
$T$ is  the  kinetic  energy of the electrons, $V$ is  the potential energy  from the electron-nuclei interaction with $v_{ext}$ being  the electrostatic  potential  of  the  nuclei \eqref{eq:coloumb_potential}; $J_H$ (the  Hartree  energy)   describes the energy corresponding to the interelectron repulsion if the electrons were mutually independent; $E_{xc}$ is the exchange-correlation energy which  accounts for correlation effects correcting the simple independent ansatz of $J_H$.

Note here that the Coulomb potential over the whole $\R^3$ is not in any $L^p$-space, but it is in $L^2(\R^3) + L^\infty(\R^3)$ and we will be using the splitting $\tfrac{1}{|\cdot|} = v_2 + v_{\infty}$, where $v_2, v_\infty$ lie in  $L^2,L^\infty$, respectively.

Precise assumptions on $e_{xc}$ which are sufficient for our mathematical results and cover standard local density approximation (LDA) exchange-correlation functionals used in practice are given in Section \ref{sec:analytic-set-up}.
A basic example derived from the homogeneous electron gas is the Dirac exchange energy
\begin{equation} \label{eq:dirac_exchange}
e_{xc}(\rho) = - \tfrac{3}{4} \left( \tfrac{3}{\pi} \right)^{\tfrac{1}{3}} \rho^{\frac{4}{3}}.
\end{equation}
The ground state of the system is given by 
\begin{equation}
    \Phi \in \argmin \E_{\mu} \text{ subject to the constraints } \eqref{eq:ortho} .
\end{equation}

Any ground state $\Phi : = (\varphi_1, \ldots, \varphi_N)$  satisfies the Euler-Lagrange equations of the system, the Kohn-Sham equations
\begin{equation} \label{eq:ks-langrange}
h_{\mu, \rho} \phi_i :=
\left( - \frac{1}{2} \Delta + v_{ext} + v_H + v_{xc} \right) \phi_i 
= \sum_{j =1}^{N} \lambda_{ij} \phi_j,
\end{equation}
where the Lagrange multipliers $\lambda_{ij}$ arise due to the orthonormality condition \eqref{eq:ortho}. The Hartree and exchange-correlation potentials are given by
\begin{equation}
    v_H(x) = \intr \frac{1}{|x-y|} \di \mu(y),
    \qquad
    v_{xc} = \drho e_{xc}.
\end{equation}{}
Since the effective one-body operator $h_{\mu, \rho}$ in \eqref{eq:ks-langrange} (the Kohn-Sham Hamiltonian) is invariant under unitary transformations, the KS equations can be brought into their canonical form 
\begin{equation} \label{eq:ks-canonical}
h_{\mu, \rho} \phi_i :=
\left( - \frac{1}{2} \Delta + v_{ext} + v_H + v_{xc} \right) \phi_i 
= \epsilon_i \phi_i.
\end{equation}

\subsection{Excitations}

Following \cite{Friesecke-kniely} we confine ourselves here to the simplest model for electronic excitations, the HOMO-LUMO transition. In this transition an electron pair migrates from the highest occupied molecular orbital (HOMO) to the lowest unoccupied molecular orbital (LUMO).
For the KS-orbitals $\Phi = (\phi_1, \ldots, \phi_n)$ ordered by the size of their eigenvalue in \eqref{eq:ks-canonical} this means
\begin{equation}
    \big( \varphi_1, \ldots, \varphi_{n-1}, \textcolor{blue}{\varphi_n} \big) 
    \enskip
    \longrightarrow
    \enskip
     \big( \varphi_1, \ldots, \varphi_{n-1}, \textcolor{red}{\varphi_{n+1}} \big),
\end{equation}
where $\phi_n$ is the HOMO and $\phi_{n+1}$ - the eigenstate corresponding to the next higher eigenvalue of $h_{\mu, \rho}$ - is the LUMO.

To define HOMO and LUMO in a variational way, we consider the \textit{excitation energy functional} \cite{Friesecke-kniely} given by the quadratic form associated with KS Hamiltonian $h_{\mu, \rho}$ \eqref{eq:ks-canonical},

\begin{equation} \label{eq:excitatin_functional}
    \E_{\mu, \rho} [\chi] 
    = \langle \chi, h_{\mu, \rho } \chi \rangle 
    = \frac{1}{2} \intr |\nabla \chi|^2 + \intr \big( v_{ext} + v_H + v_{xc}(\rho) \big) |\chi|^2.
\end{equation}

Now define a HOMO $\phi_H$ by
\begin{equation}
    \phi_H \in \argmax \E_{\mu, \rho} \text{ subject to the constraints } \phi_H \in \mathrm{Span}\{ \phi_1, \ldots, \phi_n\}, \enskip \langle \phi_H, \phi_H \rangle = 1
\end{equation}

and a LUMO $\phi_L$ by 
\begin{equation}
    \phi_L \in \argmin \E_{\mu, \rho} \text{ subject to the constraints } \langle \phi_i, \phi_L \rangle = 0,~ \forall i \in \{1, \ldots, n \} \enskip, \langle \phi_L, \phi_L \rangle = 1.
\end{equation}

If they exist, HOMO and LUMO clearly satisfy the KS equations
\begin{equation}
    h_{\mu, \rho} \phi_H = \epsilon_H \phi_H, \quad 
    h_{\mu, \rho} \phi_L = \epsilon_L \phi_L,
\end{equation}
for some eigenvalue $\epsilon_H$ (the HOMO energy) and  $\epsilon_L$ (the LUMO energy).\\

\section{Existence of HOMO-LUMO excitations \label{sec:existence}}

In this section we show that for positively charged systems $(Z>N)$ there always exist HOMO-LUMO excitations.
This generalizes a corresponding result in \cite{Friesecke-kniely} to unbounded domains, except that in bounded domains no restriction on $Z$ are needed.
In the latter case such a result is not straightforward due to the possibility of ``mass escaping to infinity'',
and requires concentration-compactness arguments \cite{Lions1, Lions2}.
The reader may wonder whether the assumption $Z>N$, which is essential in our proof, is really necessary. 
The authors of course asked themselves the same question.
For counterexamples to existence in the case $Z = N$ see Section \ref{sec:non_existence}.
\subsection{Assumptions  \label{sec:analytic-set-up}}

\textbf{Assumptions on the exchange-correlation energy}

We assume that $e_{xc}:[0,\infty) \to \R \text{ is continuously differentiable }$ with

\begin{equation} \label{eq:assumptions}
 e_{xc}(0) =0= v_{xc}(0), \enskip v_{xc} \leq 0,  \enskip |v_{xc}(\rho)| \leq c_{xc} (1 + \rho^{p-1})
\end{equation}
for some exponent $p$ with $p \in [1, \frac{5}{3})$ and constant $c_{xc} >0$.\\

Furthermore we need as in \cite{cances}:

 \begin{equation} \label{eq:assumption_at_zero}
    \text{There exists } q \in \big[1, \tfrac{3}{2}\big)  \text{ such that } \limsup_{\rho \to 0+} \frac{e_{xc}(\rho)}{\rho^{q}} < 0.
\end{equation}{}

\begin{remark}\label{rem:assumptions_dirac}
These assumptions are trivially satisfied by the Dirac exchange \eqref{eq:dirac_exchange} with $p = q = \tfrac{4}{3}$.
In our appendix we check explicitly that these assumptions are also satisfied for the two most popular LDA exchange-correlation functionals: Perdew-Zunger (PZ81) \cite{perdew-zunger} and Perdew-Wang (PW92) \cite{perdew-wang}
\end{remark}

\textbf{Admissible sets of orbitals}
In order to write these definitions in a more compact way we introduce the following sets:
The KS admissible set 
\[
\A = \{(\phi_1, \ldots, \phi_n) \in H^1(\R^3)^n : \langle \phi_i, \phi_j \rangle = \delta_{ij} \}
\]
the HOMO admissible set
\[ 
\A^{H}_{\Phi} = \big\{ \phi_H \in \mathrm{Span}\{\phi_1, \ldots, \phi_n \} : ||\phi_H||_2 = 1 \big\}
\]
and the LUMO admissible set
\[
\A_{\Phi}^{L} = \big\{ \phi_L \in H^1(\R^3) : \langle \phi_k, \phi_L\rangle =0 \text{ for } k \in \{1,\ldots,n\}, ||\phi_L||_2 =1   \big\}. 
\]

The governing variational principles for the occupied KS orbitals, HOMO and LUMO can now be summarized as 

\begin{equation}\label{eq:definition_homo_lumo}
    \Phi \in \argmin_{\A} \E_{\mu}, 
    \quad
    \phi_H \in \argmax_{\A_{\Phi}^{H}}\E_{\mu, \rho},
    \quad
    \phi_L \in \argmin_{\A_{\Phi}^{L}} \E_{\mu, \rho}.
\end{equation}

We start with some estimates for the KS energy functional \eqref{eq:ks-energy-functional}.
\begin{lem}[Lower bounds on KS energy functional]\label{lem:bounds_ks}
The terms in the KS energy functionals have the following properties
\begin{enumerate}
    \item $T[\Phi] \geq \tfrac{1}{2} T[\Phi] + \tfrac{1}{c} ||\rho||_{3}$ and   
    \[ 
    \Phi \mapsto T[\Phi] \text{ is continuous and weakly lower semicontinuous on } H^1(\R^3)^n.
    \]
    \item $V[\rho] \geq - ||\mu||_{\mathcal{M}} \big(||v_\infty||_\infty ||\rho||_1 + ||v_2||_2 ||\rho||^{\nicefrac{1}{4}}_1 ||\rho||_{3}^{\nicefrac{3}{4}}  \big)$ and 
    \[ (\Phi, \mu) \mapsto V[\rho] \text{ is strong $\times$ weak$^{*}$ continuous on } \big(L^{4} \cap L^2 \big)\times \mathcal{M} .
    \]
    \item $J_H \geq 0$ and 
    \[
    \Phi \mapsto J_H[\rho] \text{ continuous on } \big(L^{\nicefrac{12}{5}}(\R^3)\big)^n .
    \]
    \item $E_{xc}[\rho] \geq - c_{xc} \big( ||\rho||_1 + \frac{1}{p-1} ||\rho||_1^{\nicefrac{(3-p)}{2}} ||\rho||_{3}^{\nicefrac{3(p-1)}{2}} \big),$ where $p \in (1, \tfrac{5}{3})$ is the exponent from assumption \eqref{eq:assumptions}, and 
    \[
    \Phi \mapsto E_{xc}[\rho] \text{ is continuous on } \big(L^{2p}(\R^3)\big)^n.
    \]
\end{enumerate}

\end{lem}

\begin{proof}
The first inequality is a standard result in DFT, but we include it for the sake of completeness.
By a well-known result, see e.g. \cite{catto_lebris_lions}, we have $T[\Phi] \geq ||\nabla \sqrt{\rho}||^{2}_{2}$,  so the inequality follows by applying the Sobolev embedding $H^1 \hookrightarrow L^6$ to the function $u = \sqrt{\rho}$.

Estimate 2 follows from the duality between $\M(\Omega_{nuc})$ and $C_b(\Omega_{nuc})$ and then Cauchy-Schwarz
\[
V[\rho] = - \intr \left( \frac{1}{|\cdot|} \ast \rho\right) \di \mu 
\geq 
- ||\mu||_{\M}  ||\tfrac{1}{|\cdot|} \ast \rho||_{\infty} 
\geq 
- ||\mu||_{\M}  \big( ||v_2||_{2} ||\rho||_{2} + ||v_{\infty}||_{\infty}  ||\rho||_{1} \big)
\]
and finally bounding the $L^2$-norm of $\rho$ by the H\"{o}lder interpolation inequality,

\begin{equation} \label{eq:hoelder-int}
||\rho||_{p} \leq ||\rho||_{q}^{\theta} ||\rho||_{r}^{1-\theta} 
\text{ with } q \leq p \leq r, \quad  \tfrac{1}{p} = \tfrac{\theta}{q} + \tfrac{1- \theta}{r} 
\end{equation}
with $p=2, q=1,r=3$.

The positivity of $J_H$ is trivial. \\
Ad 4:
By our assumtpion on $v_{xc}$ we have
\[
| e_{xc}(\rho)| = \big| e_{xc}(0) + \int_{0}^{\rho} v_{xc}(\xi) \di \xi \big|
\leq
c_{xc} \big(\rho + \tfrac{1}{p-1} \rho^p \big)
\]
The estimate again follows from H\"{o}lder interpolation \eqref{eq:hoelder-int} with $q=1, r = 3, \theta = \tfrac{(3-p)}{2p}$.
In all four cases the continuity results follow by pointwise continuity of the integrand and the proven bounds.
\end{proof}

The second Lemma considers the excitation functional \eqref{eq:excitatin_functional}. In the following we denote $\rho_\chi = |\chi|^2$.

\begin{lem}[Lower bounds on excitation functional]\label{lem:bounds_excitation_energy}
The terms in the excitation functional have the following properties 
\begin{enumerate}
    \item $T[\psi] \geq \tfrac{1}{2} T[\psi] + \frac{1}{c} ||\rho_\psi||_3$ and 
    \[
    \chi \mapsto T[\chi] \text{ is coninuous and weakly lower semicontinuous on } H^1(\R^3).
    \]
    \item $\intr v_{ext} \rho_\psi \geq  - ||\mu||_{\mathcal{M}} \big(||v_\infty||_\infty ||\rho_\psi||_1 + ||v_2||_2 ||\rho_\psi||^{\nicefrac{1}{4}}_1 ||\rho_\psi||_{3}^{\nicefrac{3}{4}}  \big)$ and 
    \[ (\chi, \mu) \mapsto \intr v_{ext} \rho_{\chi} \text{ is strong $\times$ weak$^{*}$ continuous on } \big(L^{4} \cap L^2 \big)\times \mathcal{M} .
    \]
    \item $\intr (\tfrac{1}{|\cdot|} \ast \rho) \rho_{\psi} \geq 0$ and 
    \[
    (\Phi, \chi) \mapsto \intr \big(\frac{1}{|\cdot| } \ast \rho \big) \rho_{\chi} \text{ is continuous on } L^{\nicefrac{12}{5}}(\R^3)^{n+1}.
    \]
    \item $\intr v_{xc}(\rho) \rho_\psi \geq - c_{xc} \big( ||\rho_\psi||_1 + ||\rho||_{p}^{p-1} ||\rho_\psi||_{1}^{\nicefrac{(3-p)}{(2p)}} ||\rho_\psi||_{3}^{\nicefrac{3(p-1)}{(2p)}}\big)$ where $p \in [1,\tfrac{5}{3})$ is again the exponent from our assumptions on $v_{xc}$  and 
    \[
    (\Phi, \chi) \mapsto \intr v_{xc}(\rho) \rho_{\chi} \text{ is continuous on } L^{2p}(\R^3)^n \times \big(  L^2(\R^3) \cap L^{2p}(\R^3) \big).
    \]
\end{enumerate}
In particular, the map $(\Phi, \chi, \mu) \mapsto \E_{\mu, \rho}[\chi]$ is weak $\times$ strong $\times$ weak$^*$ continuous and weak $\times$ weak $\times$ weak$^*$ lower semicontinuous on $ H^1(\R^3)^n \times H^1(\R^3) \times \mathcal{M}$.
\end{lem}

\begin{proof}
Statements 1-3 follow by the same line of reasoning as in Lemma \ref{lem:bounds_excitation_energy} above. The fourth assertion follows by the same argument given in \cite{Friesecke-kniely}, but we include it for the sake of completeness.
By our assumption on the exchange-correlation potential $v_{xc}$ we have
\[
\intr v_{xc}(\rho) \rho_\psi 
\geq 
- c_{xc}\intr \big( 1 + \rho^{p-1}\big)\rho_\psi 
\geq - c_{xc} \left( ||\rho_\psi||_1 + \underbrace{||\rho^{p-1} ||_{p'}}_{ = ||\rho||_{p}^{p-1}} \cdot ||\rho_\psi||^{p}\right)
\quad 
\text{with } p' = \frac{p}{p-1}.
\]
The asserted bound now follows from the H\"older interpolation inequality \eqref{eq:hoelder-int} with $  q = 1, r = 3, \theta = \frac{3-p}{2p}$.
The continuity follows from the pointwise continuity of $\rho \mapsto v_{xc}(\rho)$ together with the bounds \eqref{eq:assumptions}.
\end{proof}

\begin{remark}\label{rem:assumptions}
Our existence results do not require $\Phi$ to be the KS ground state, but only to satisfy a fast enough decay property, like

\begin{equation} \label{eq:decay}
    \exists \gamma >0 ~~ s.t \quad e^{\gamma |\cdot|} \varphi_j \in H^1(\R^3), \quad \forall j \in \{1, \ldots, n\}.
\end{equation}
\end{remark}
By the results of \cite{cances}, whenever $Z \geq N=  2n$ and $e_{xc}$ satisfies our assumptions \eqref{eq:assumptions}, then there exists a KS ground state $\Phi = (\varphi_1, \ldots, \varphi_n)$
and any such ground state fulfills the additional property \eqref{eq:decay}.

\begin{theorem}[Existence of HOMO-LUMO excitations]\label{thm:lumo}
For any admissible nuclear charge distribution $\mu \in \mathcal{A}_{nuc}$ and any set of orbitals $\Phi = (\phi_1, \ldots, \phi_n) \in \mathcal{A}$ the excitation functional possesses a maximizer $\phi_H$ on $\mathcal{A}_{\Phi}^H$ (i.e., a HOMO).

If additionally we have $Z >N$, i.e. a positively charged system, and $\Phi \in \A$ satisfies the decay property \eqref{eq:decay}, then there exists also a minimizer $\phi_L$ on $\mathcal{A}_{\Phi}^L$ (i.e., a LUMO).
\end{theorem}

\begin{proof}
Existence of a HOMO is elementary since $\chi \mapsto \E_{\mu, \rho}[\chi]$ is continuous on $H^1(\R^3)$ (see Lemma \ref{lem:bounds_excitation_energy}) and the admissible set $A_{\Phi}^{H}$ is compact since it is a closed and bounded subset of a finite dimensional space.
The case of a LUMO is more difficult and will be treated in the next section.
\end{proof}

\subsection{Existence proof}

For a given $\Phi = (\phi_1, \ldots, \phi_n)$, define $\gamma_\Phi := \sum_{k=1}^{n} |\phi_k \rangle \langle \phi_k|$, the projection onto the span of the $\phi_k$. 
Then we define the translation invariant excitation functional by 
\begin{equation}
    \E[\chi] := 
    \begin{cases}
    \E_{\mu, \rho} \bigg[  \frac{(Id - \gamma_\Phi) \chi}{||(Id - \gamma_\Phi) \chi||}\bigg] = \frac{1}{||(Id - \gamma_\Phi) \chi||^2} \E_{\mu, \rho}[(Id - \gamma_\Phi) \chi],
    & \text{for } \chi \not \in \mathrm{span} \{\phi_1, \ldots, \phi_n \},  \\
    0 
    & \text{otherwise}.
    \end{cases}
\end{equation}

Note $\E$ depends on $\mu, \Phi$ but to keep notation simple and distinguish $\E$ from $\E_{\mu, \rho}$ we do not explicitly stress this via the notation. 

Now $\E$ is again defined over the whole $H^1(\R^3)$ and if we can prove the existence of a minimizer $\hat{\chi} $ with 
\[
\hat{\chi} := \text{argmin } \E[\chi], \quad  \text{subject to } \chi \in H^1(\R^3), \enskip ||(Id - \gamma_\Phi)\chi|| = 1,
\]
then we have the existence of a LUMO given by $\phi_L := (Id - \gamma_\Phi)\hat{\chi}$.

As usual in a concentration compactness argument  we embed our problem in a one-parameter family of minimization problems
\begin{equation}\label{minimization_projection}
I_\lambda := \inf \limits_{\chi \in K_\lambda} \E[\chi], \qquad K_\lambda := \{\psi \in  H^1(\R^3) : ||(Id-\gamma_\Phi)\psi||_{2}^2 = \lambda \}
\end{equation}
parametrized by $\lambda \in \R_+$.

We can now state the first central Lemma of this section:

\begin{lem}[Fundamental properties of the infimum]\label{lem:fundamental_properties_infimum}
The map given by $\lambda \mapsto I_\lambda$ has the following properties
\begin{enumerate}[(i)]
    \item $I_0 =  0$ and $-\infty < I_\lambda < 0$ for every $\lambda >0$,
    \item $\lambda \mapsto I_\lambda$ is continuous and strictly decreasing, i.e. for $0<\eta<\lambda$ the inequality $I_\lambda < I_\eta $ holds
\end{enumerate}
\end{lem}

\begin{proof}
The fact that $I_0 =0$ is clear and since $(Id- \gamma_\Phi)$ is a bounded operator from $H^1$ into itself $I_\lambda > -\infty$ follows by the bounds in Lemma \ref{lem:bounds_excitation_energy}.

Usually the strict subadditivity condition is the hardest part in a concentration compactness proof, but fortunately here it is quite easy.

The only thing we have to prove is actually $I_\lambda <0$ for $\lambda >0$. Since by the structure of our problem we have $\theta I_\eta = I_{\theta \eta}$ for all $\theta, \eta \geq 0$. Hence, for $\lambda > \eta$ we choose $\theta >1$ such that $\theta \eta = \lambda$ and obtain
\[
I_\lambda = I_{\theta \eta} = \theta I_\eta < I_\eta .
\]
Also continuity is clear since $\lambda \mapsto I_\lambda$ is just a linear function.
The inequality $I_\lambda <0$ is also important because otherwise we would be considering the constant zero-function.

$\boxed{I_\lambda <0}$
Take a radially symmetric testfunction $\psi \in C^{\infty}_c(\R^3)$ with supp$\psi \subseteq B_K^c$ for some radius $K >0$, $||\psi|| = 1$. Here $B_K^c$ denotes the complement of the ball $B_K$ around the origin of radius $K>0$.
Then define 
\[
\psi_\sigma (x) = \sigma^{\nicefrac{3}{2}} \psi \big( \sigma (x-\tfrac{1}{\sqrt{\sigma}} \hat{\e})\big), \qquad \text{with }\hat{\e} \text{ some unitvector in } \R^3, \sigma >0.
\]
Note that 
\begin{align*}
    \langle \phi_k, \psi_\sigma \rangle 
    =
    \intr \phi_k(x) \sigma^{\nicefrac{3}{2}} \psi \big( \sigma (x-\tfrac{1}{\sqrt{\sigma}} \hat{\e})\big) \di x
    =
    \sigma^{- \nicefrac{3}{2}} \int_{B_K^c} \psi(y) \phi_k\big(\tfrac{y}{\sigma} + \tfrac{1}{\sqrt{\sigma}} \eh \big) \di y = O(\sigma^{- \nicefrac{3}{2}} \exp(- \tfrac{c}{\sigma})),
\end{align*}
where we used in the second to last equality that supp$\psi \subseteq B_K^c$ and in the last step the exponential decay of the KS-orbitals together with the estimate
\[
\big| \tfrac{y}{\sigma} + \tfrac{1}{\sqrt{\sigma}} \eh \big| 
\geq
\left| \big| \tfrac{y}{\sigma}\big| -  \big| \tfrac{1}{\sqrt{\sigma}} \eh \big| \right|
\geq \tfrac{c}{\sigma}, \quad \text{for } y \in B_K^c \text{ and $\sigma$ small enough and some constant $c>0$. } 
\]
Hence $ \langle \phi_k, \psi_\sigma \rangle $ decays exponentially for $\sigma \to 0$, so it can be neglected up to higher orders, i.e.
\[
\big| (Id- \gamma_\Phi) \psi_\sigma\big|^2 =  |\psi_\sigma|^2 + O(\exp(-\tfrac{c}{\sigma}) )
\quad \text{and} \quad
\big| \nabla (Id- \gamma_\Phi) \psi_\sigma\big|^2 =  |\nabla \psi_\sigma|^2 + O(\exp(-\tfrac{c}{\sigma})) .
\]

With this at hand we can estimate the energy of $\psi_\sigma$
\begin{align*}
    \E[\psi_\sigma]
    =
    \frac{1}{2} \intr \big| \nabla \psi_\sigma(x) \big|^2 \di x
    +
    \intr \big( v_{xc} + v_{ext} + v_{H} \big) |\psi_\sigma (x)|^2 \di x
    + O(\exp(-\tfrac{c}{\sigma})).
\end{align*}

Since we want to estimate the energy from above we do not need to consider the exchange-correlation term since it only gives a negative contribution.
The kinetic energy is easily estimated 
\[
2 T[\psi_\sigma] 
=
\intr 
|\nabla \psi_\sigma|^2 \di x
=
\sigma^2 \intr |\nabla \psi|^2 \di x
= O(\sigma^2).
\]

Our next step will be to estimate the Hartree and the nuclei term and bring them in a similar form
\begin{align*}
    \intr v_H |\psi_\sigma|^2 \di x
    &=
    \intr \intr \frac{\rho(x) |\psi_\sigma(y)|^2}{|x-y|} \di x \di y
    =
    \intr \intr \frac{\rho(x) |\psi(z)|^2}{\big|x-(\tfrac{z}{\sigma} + \tfrac{1}{\sqrt{\sigma}} \eh )\big|} \di x \di z \\
    &=
    \sigma \intr \rho(x) \intr \frac{|\psi(z)|^2}{\max \{ \sigma |x|, |z + \sqrt{\sigma} \eh | \}} \di z \di x,
\end{align*}
where we used the radial symmetry of $\psi$ in the last step.

Exactly the same steps transform the external potential term into
\[
\intr  v_{ext} |\psi_\sigma|^2 \di x 
=
 -\sigma  \intr \di \mu(x) \intr \frac{|\psi(z)|^2}{\max \{ \sigma |x|, |z + \sqrt{\sigma} \eh | \}} \di z.
\]

Since the support of $\mu$ is compact and again due to supp$\psi \subseteq B_K^c$, we have for $\sigma$ small enough  in the last integral $\max \{ \sigma |x|, |z + \sqrt{\sigma} \eh | \} = |z + \sqrt{\sigma} \eh | $.

Putting everything together we obtain for $\sigma$ small enough
\begin{align*}
I_{||(Id - \gamma_\Phi)\psi_\sigma||}
&\leq 
\E[\psi_\sigma]
\leq 
\intr \big( v_{H} + v_{ext}\big)| \psi_\sigma|^2 \di x + O(\sigma^2) \\
&=
\sigma 
\left(- Z \intr \frac{|\psi(z)|^2}{ |z + \sqrt{\sigma} \eh | } \di z +  \intr \rho(x) \intr \frac{|\psi(z)|^2}{\max \{ \sigma |x|, |z + \sqrt{\sigma} \eh | \}} \di z \di x \right) + O(\sigma^2)\\
&\leq 
\sigma (2n - Z) \intr \frac{|\psi(z)|^2}{ |z + \sqrt{\sigma} \eh | } \di z + O(\sigma^2) < 0.
\end{align*}
In this last inequality the assumption of a positively charged system $Z>N$ is crucial.

\end{proof}

Our last Lemma is just a standard lower semicontinuity result.

\begin{lem}\label{lowersemicontinuity}
Let $(\chi_n)_{n \geq 1}$ be bounded in $H^1(\R^3)$ converging weakly to $\chi$. If $||\chi||_2 = \limit{n}{\infty} ||\chi_n||_2,$ then $\chi_n \to \chi $ strongly in $L^p(\R^3)$ for $p \in [2,6)$ and furthermore
\[
\E[\chi] \leq \liminf \limits_{n \to \infty} \E[\chi_n].
\]
\end{lem}

\begin{proof}
The weak convergence in $H^1$ together with the convergence of the $L^2$-norms gives strong convergence in $L^2$ which in turn gives strong convergence in $L^p$, $p\in[2,6)$ due to H\"{o}lder interpolation inequality and the fact that $(\chi_n)_{n\geq1}$ is bounded in $L^q$ for  $q\in[2,6]$ by Sobolev embedding.

The inequality follows by weak lower semicontinuity of $\chi \mapsto \E_{\mu, \rho}[\chi]$ established in Lemma \ref{lem:bounds_excitation_energy} and the fact that $ (Id-\gamma_\Phi)$ is a continuous operator from $H^1$ into itself. 
\end{proof}

With all these Lemmata at our disposal, we are now able to prove the existence of a minimizer to \eqref{minimization_projection} and hence the existence of a LUMO.

Take a minimizing sequence $(\chi_n)_{n}$ for $I_\lambda$. Then w.l.o.g. $\chi_n = (Id-\gamma_\Phi){\chi}_n$, i.e. $\chi_n$ lies in the orthorgonal complement of  $\mathrm{span}\{\phi_1,\ldots,\phi_n\}$. 
Therefore the sequence is bounded in $H^1(\R^3)$ by Lemma \ref{lem:bounds_excitation_energy} and hence (up to taking a subsequence) $\chi_n \rightharpoonup \chi$ weakly in $H^1$ for some $\chi \in H^1$.

If $||\chi||_2 = \limit{n}{\infty} ||\chi_n||_2$, then by Lemma \ref{lowersemicontinuity} we get
\[ \E[\chi] \leq \liminf \limits_{n\to\infty} \E[\chi_n] = I_\lambda,\]
and since $\chi$ also lies in the (convex and closed, hence weakly closed) orthorgonal complement of $\mathrm{Span}\{ \phi_1, \ldots, \phi_n\}$ we have $||(Id-\gamma_\Phi)\chi|| = ||\chi|| = \lambda$ so $\chi$ is a minimizer.

The rest of the proof consists in ruling out $||\chi|| < \limit{n}{\infty} ||\chi_n|| =: \Lambda$. 

For this purpose, assume first $||\chi|| =0$. Then $\chi_n \to 0$ and $(Id- \gamma_\Phi)\chi_n \to 0$ in $L^{p}_{loc}(\R^3)$ for $p \in [2,6)$. 
In particular,
\[ 
I_\lambda 
= \limn \E[\chi_n] 
\geq
\liminf\limits_{n\to \infty} 
\int_{B_K} \big( v_{xc} + v_{ext}\big) |(Id- \gamma_\Phi)\chi_n |^2 + \int_{B_K^c} \big( v_{xc} + v_{ext}\big) |(Id- \gamma_\Phi)\chi_n |^2 ,
\]
for any fixed $K>0$.

The first term converges to 0 and the second one can be made arbitrarily small because $v_{ext}$ vanishes at infinity and $v_{xc}(\rho)$ also vanishes since $\rho$ decays exponentially. 

So in conclusion we would get $I_\lambda \geq 0$, a contradiction to Lemma \ref{lem:fundamental_properties_infimum}.

So we only have to rule out $\alpha := ||\chi|| \in (0,\Lambda) $. Here we use a standard argument. Following e.g. \cite{frank2007muller}, define a quadratic partition of unity $\xi^2 + \zeta^2 =1$ with $\xi$ being smooth, radially non-increasing and with 
\[
\xi(0) = 1, ~ 0\leq \xi < 1 \text{ for } |x| \neq 0, ~ \xi = 0 \text{ for } |x| \geq 1, ~||\nabla \xi||_{\infty} \leq 2, ~ ||\nabla \zeta||_{\infty} \leq 2. 
\]
Define $\xi_R := \xi(\tfrac{\cdot}{R})$. Then for all $n \in \N,$ the map $R \mapsto ||\chi_n\xi_R||$ is continuous, non-decreaing, vanishes for $R\to 0$, and converges to $\lambda$ for $R \to \infty$.

Hence we can choose a sequence $(R_n)_{n \geq 1}$ such that $||\chi_n \xi_{R_n}|| = \alpha$.
Furthermore we claim that this sequence goes to infinity as $R_n \to \infty$. 

Assume for contradiction that this were not the case. Then up to a subsequence we have $R_n \to R_*$ for some $R_* \geq 0$. But this would give us
\[
\intr |\chi|^2 \xi^2_{R_*} \di x = \limn \intr |\chi_n|^2 \xi^2_{R_n} \di x = \alpha = \intr |\chi|^2 \di x,
\]
which is not possible since $\xi^2_{R_*} <1 $ for $x \neq 0$. Hence $R_n \to \infty$.

Next we define $\chi_{1,n} := \chi_n \xi_{R_n}$ and $\chi_{2,n} := \chi_n \zeta_{R_n} $ such that $|\chi|^2 = |\chi_{1,n}|^2 + |\chi_{2,n}|^2$ and (by definition of the $R_n$) $||\chi_{1,n}|| = \alpha, $ $||\chi_{2,n}|| = ||\chi_n||^2 - \alpha$.

Note that those two sequences are also bounded in $H^1(\R^3)$, since $\chi_n$ is bounded in $H^1$ and $\xi_{R_n}$ belongs to $W^{1,\infty}$.

Furthermore the cutoff gives us that $(\chi_{1,n})_n$ converges weakly to $\chi$ and hence as a by-product $\chi_{2,n}$ converges to 0 in $L^p_{loc}$. In order to see this, we take a $\Psi \in C^{\infty}_{c}$ and consider
\[
\langle \Psi, \chi_{1,n} \rangle 
=
\intr \Psi \xi_{R_n} \chi_n \di x 
=
\intr \Psi \xi_{R_n} \di x 
+
\intr \Psi \xi_{R_n} ( \chi_n -1) \di x
\gegen{n}{\infty}
\intr \Psi \chi \di x,
\]
where we used that $\chi_n$ converges weakly to $\chi$ and the fact that $\xi_{R_n}$ converges pointwise to 1 and stays uniformly bounded. 
So we proved weak convergence on a dense subset and by boundedness we conclude.

Now we can apply Lemma \ref{lowersemicontinuity} to the sequence $\chi_{1,n}$ since by construction $||\chi_{1,n}|| = \alpha = ||\chi||$, and obtain $\E[\chi] \leq \liminf \limits_{n \to \infty} \E[\chi_{1,n}]$ and $\chi_{1,n} \to \chi$ in $L^{p}$ for $p \in [2,6)$.

Our next step is proving 
\[
\E[\chi_n] 
\geq \E[\chi_{1,n}] + o(1),
\]
where $o(1)$ is the usual Landau notation meaning terms which vanish for $n \to \infty$.

First note that $\chi_{1,n} \rightharpoonup \chi$ and also $\langle \phi_k, \chi \rangle =0 $, hence
\[
\langle \phi_k, \chi_{1,n} \rangle = \langle \phi_k, \chi \rangle + o(1) = o(1).
\]
Even more, remembering $R_n \to \infty$ for $n \to \infty$, we have $|\nabla \xi_{R_n}| \leq \tfrac{c}{R_n} = o(1)$. 
Therefore we obtain
\[
\intr |\nabla \chi_{1,n} |^2 \di x 
=
\intr \big| \xi_{R_n} \nabla \chi_{n}
+
\chi_n \underbrace{\nabla \xi_{R_n} }_{= o(1)}\big|^2 \di x 
=
\intr \xi^{2}_{R_n} |\nabla \chi_{n}|^2 \di x + o(1)
\leq 
\intr |\nabla \chi_n|^2 \di x+ o(1).
\]
This already gives us $T[\chi_{1.n}] \leq T[\chi_{n}] + o(1)$.

For the potential energy we start by proving $V[\chi_n] = V[\chi_{1,n}] + V[\chi_{2,n}] + o(1)$.

Due to the fact that $(\chi_{2,n})_n$ is bounded and converges pointwise to zero we have $\chi_{2,n} \rightharpoonup 0$ in $H^1(\R^3)$ and hence $ \langle \phi_k, \chi_{2,n} \rangle = o(1)$, which gives us $(Id-\gamma_\Phi)\chi_{2,n} = \chi_{2,n} + o(1)$. Consequently 
\begin{align*}
V[\chi_{1,n}] + V[\chi_{2,n}]
&= 
\intr \big( v_{ext} + v_{H} + v_{xc} \big) 
\left( \big| (Id - \gamma_\Phi)\chi_{1,n}\big|^2
+ \big| (Id - \gamma_\Phi)\chi_{2,n}\big|^2 \right) \\
&=
\intr \big( v_{ext} + v_{H} + v_{xc} \big) |\chi_n|^2
+o(1).
\end{align*}
So up until now we have
\[
\E[\chi_n] 
= T[\chi_n] + V[\chi_n]
\geq 
T[\chi_{1,n}]
+V[\chi_{1,n}] + V[\chi_{2,n}] + o(1)
= \E[\chi_{1,n}] + V[\chi_{2,n}] + o(1).
\]
The only fact that remains to be proven is $V[\chi_{2,n}] \geq  o(1)$. This follows from the standard argument that the only negative contributions come from $v_{ext}$ and $v_{xc}$, i.e.
\[
\liminf \limits_{n \to \infty} 
V[\chi_{2,n}]
\geq 
\liminf \limits_{n \to \infty} 
\int_{B_K} \big( v_{xc} + v_{ext}\big) |\chi_n |^2 + \int_{B_K^c} \big( v_{xc} + v_{ext}\big) |\chi_n |^2, 
\]
where the first term vanishes because $ \chi_{2,n} \to 0$ in $L^{p}_{loc}$, and the second one also tends to $0$ because $K>0$ was arbitrary and $v_{xc}$ and $v_{ext}$ decay at infinity.

In conclusion we proved $\E[\chi_n] \geq \E[\chi_{1,n}] + o(1)$ which gives us after passing to the limit
\[
I_{\lambda}
=
\limn \E[\chi_n]
\geq 
\liminf \limits_{n \to \infty}
\E[\chi_{1,n}]
\geq 
\E[\chi]
\geq I_\alpha,
\]
 where we used in the last inequality that $\chi$ lies in the orthogonal complement of the occupied orbitals and hence $\alpha = ||\chi|| = ||(Id - \gamma_\Phi)\chi||$ so $\chi \in K_\alpha$.
 
 But since we assumed $\alpha < \lambda$, this contradicts statement $(ii)$ in Lemma \ref{lem:fundamental_properties_infimum}.
This completes the proof of Theorem \ref{thm:lumo}.

\subsection{Higher excitations}
We now complete the rigorous justification of the picture on the left-hand side of Figure \ref{fig:spectra}, i.e. we fully characterize the spectrum of the KS Hamiltonian in the positively charged case.

\begin{theorem}[Spectrum of KS Hamiltonian in the case $Z>N$] \label{thm:neutral_case}
Consider $Z>N$, i.e. a positively charged system. Then the KS Hamiltonian $h_{\mu, \rho}$ given in \eqref{eq:ks-canonical} has infinitely many negative eigenvalues of finite multiplicity below the bottom of its essential spectrum,
\begin{equation}
    \sigma \big( h_{\mu, \rho}\big) = \{\lambda_n\}_{n \geq 1} \cup [0, \infty), \quad 
    \text{ with }
    \lambda_n <0 \text{ and } \lambda_n \gegen{n}{ \infty} 0.
\end{equation}
\end{theorem}
\begin{proof}
Due to assumption \eqref{eq:assumptions} the potential $v_{ext} + v_H + v_{xc}$ belongs to  the space $L^2(\R^3) + L^{\infty}_{\epsilon}(\R^3)$, consisting of potentials $v$ which for any given $\epsilon>0$, can be decomposed as $v = v_1 +v_2$ with $||v_1||_2 < \infty$, $||v_2||_{\infty} < \epsilon$. 
Therefore is relatively compact w.r.t. the Laplacian.
Hence by Weyl's Theorem (see e.g. Chapter XIII.4 of \cite{reed1978methods}) $h_{\mu, \rho} $ is self-adjoint and we have $\sigma_{ess} \big(h_{\mu, \rho} \big) =\sigma_{ess} \big( - \Delta \big)  = [0, \infty)$.
Furthermore, $h_{\mu, \rho}$ is bounded from below. To see this, write $h_{\mu, \rho} = - \Delta + U_2 + U_\infty$ with $U_2 \in L^2(\R^3)$ and $U_\infty \in L^{\infty}_{\epsilon}(\R^3)$. Then for every $\psi \in H^1(\R^3)$ with $||\psi||_2 =1$ we have
\begin{align}
    \langle \psi, h_{\mu, \rho} \psi \rangle
    \geq 
    \frac{1}{2} ||\nabla \psi||_{2}^{2} - ||U_2||_2 ||\psi||_{4}^{2} - ||U_{\infty}||_{\infty}.
    \end{align}
Now the Sobolev embedding $||\phi||_{4}^{2} \leq C ||\nabla \phi||_{2}^{\nicefrac{3}{2}}$ and Young's inequality give the asserted lower bound.
Now note that
\begin{equation} \label{eq:hamiltonian_inequality}
    h_{\mu, \rho} \leq - \frac{1}{2} \Delta + v_{ext} + v_{H}
\end{equation}
and since by assumption $Z>N$ we can apply Lemma II.1 of \cite{Lions-Hartree} which gives us that the right-hand side operator of \eqref{eq:hamiltonian_inequality} is negative on an infinite-dimensional subspace. Hence so is $h_{\mu, \rho}$. The min-max principle now gives that it has infinitely many negative eigenvalues.
This completes the proof.
\end{proof}

\section{Optimal HOMO-LUMO excitations \label{sec:excitation_sets}}
In \cite{Friesecke-kniely}, motivated by the design of photovoltaic materials, results are given for \textit{optimal} HOMO-LUMO excitations with respect to some control goals. The control is the nuclear charge distribution (e.g., the doping profile or the heteroatom substitutions); typical control goals are the spatial electron-hole charge separation or the size of the HOMO-LUMO gap.
The analysis in \cite{Friesecke-kniely} relies on the simplyfing assumption of bounded domains. Here we generalize this analysis to unbounded domains.
The main difficulty again consists in handling loss of mass at infinity.

\begin{lem}[Analytic properties of the set of HOMO-LUMO excitations] \label{lem:compactness}
Consider a positively charged system, i.e. $Z>N$, then
the joint solution set to the governing variational pronciples for occupied KS orbitals, HOMO and LUMO parametrized by the set of nuclear charge distributions $\mu$,
\[
\mathcal{B}
= \{ (\Phi, \phi_H, \phi_L, \mu) : \mu \in \A_{nuc},~ (\Phi,\phi_H,\phi_L) \text{ definded by } \eqref{eq:definition_homo_lumo}\}
\]
has the following properties:
\begin{enumerate}[(a)]
    \item $B$ is weak $\times$ weak $\times$ weak $\times$ weak$^*$-closed in $H^1(\R^3)^n \times H^1(\R^3) \times H^1(\R^3) \times \M$
    \item $B$ is strong $\times$ strong $\times$ strong $\times$ weak$^*$-compact in $H^1(\R^3)^n \times H^1(\R^3) \times H^1(\R^3) \times \M$
\end{enumerate}
\end{lem}


\begin{proof}
\textbf{Part (a)}
Let $(\Phi^{(\nu)},\phi_H^{(\nu)},\phi_L^{(\nu)}) \rightharpoonup (\Phi, \phi_H,\phi_L)$ in $H^1(\R^3)^{(n+2)}$ and $\mu^{(\nu)} \rightharpoonup^{*} \mu$ in $\M$, then we need to prove
\begin{align*}
\text{(i) } \mu \in \A_{nuc} \quad 
\text{(ii) } \Phi \in \argmin_{\A} \E_{\mu} \quad
\text{(iii) } \phi_H \in \argmax_{\A_{\Phi}^{H}}\E_{\mu, \rho} \quad
\text{(iv) } \phi_L \in \argmin_{\A_{\Phi}^{L}} \E_{\mu, \rho} .
\end{align*}

\textbf{Ad (i):} Since all measures $\mu^{(\nu)}$ are supported on the compact set $\Omega_{nuc}$, the constant functions are in the predual of $\M$, that is the space of the continuous functions on $\Omega_{nuc}$, and hence $ Z = \mu^{(\nu)} \left( \Omega_{nuc}\right) \to \mu \left( \Omega_{nuc}\right) $, so $\mu \in A_{nuc}.$

\textbf{Ad (ii): }
For any admissible $\Psi \in \A$ we have by the variational definition of $\Phi^{(\nu)}$ and the weak$^*$  continuity of $\mu \mapsto \E_{\mu}[\Psi]$
\begin{equation} \label{eq:beginning_argument}
\E_{\mu^{(\nu)}}[\Phi^{(\nu)}] 
\leq 
\E_{\mu^{(\nu)}}[\Psi]
 \gegen{\nu}{\infty}
 \E_{\mu}[\Psi],
\end{equation}

which implies
\begin{equation} \label{eq:inf_upper}
\limsup \limits_{\nu \to \infty} \E_{\mu^{(\nu)}}[\Phi^{(\nu)}]  \leq \inf \limits_{\A} \E_{\mu},
\end{equation}

since $\Psi \in \A$ was arbitrary.

Using the weak $\times$ weak$^*$ lower semicontinuity of $(\Phi, \mu) \mapsto \E_{\mu}[\Phi]$ on $H^1(\R^3)^n \times \M$ gives 
\begin{equation} \label{eq:inf_lower}
\liminf \limits_{\nu \to \infty} \E_{\mu^{(\nu)}}[\Phi^{(\nu)}]
\geq 
\E_{\mu}[\Phi],
\end{equation}

but unlike in the bounded domain case this does not directly give us the result since $\A$ is not weakly closed.

Fortunately, the $\Phi^{(\nu)}$ are not arbitrary elements of $\A$. Since the energy functional is invariant under unitary transformations, we can always assume the orbitals $\Phi = (\phi_1, \ldots, \phi_n)$ to be orthogonal. 
So the only thing that could go wrong is loss of mass due to weak convergence. 

Assume for contradiction $\alpha = ||\Phi|| < n$ so we lose mass in at least one orbital $\phi_k$.
In order to see that this is not possible, we will place some small mass at a large but finite distance, and obtain a state with lower energy.

We take $\eta \in C^{\infty}_c(\R^3)$ with $||\eta||_{2} = 1$ and consider the test function $\eta_{\lambda, \sigma} := \sigma^{\nicefrac{1}{2}} \lambda^{\nicefrac{3}{2}} \eta(\lambda \cdot)$.

Then
\begin{align*}
    \E_{\mu}[\eta_{\lambda, \sigma}]
    =
    \sigma \lambda^2 T[\eta]
    +
    V[|\eta_{\lambda, \sigma}|^2]
    +
    \sigma^{2} \lambda J[|\eta|^2]
    +
    E_{xc}[|\eta_{\lambda, \sigma}|^2].
\end{align*}
Using the assumption \eqref{eq:assumption_at_zero} and the fact that $V[|\eta|^2] \leq 0$, we obtain that for $\sigma$ small enough there exists a constant $c>0$ such that

\begin{align*}
    \E_{\mu}[\eta_{\lambda, \sigma}]
    \leq 
    \sigma \lambda^2 T[\eta]
    +
    \sigma^{2} \lambda J[|\eta|^2]
    -
    c \sigma^{q} \lambda^{3(q-1)} \intr |\eta|^{2q} \di x.
\end{align*}
Since $q < \tfrac{3}{2}$, the negative term dominates in the limit $\sigma, \lambda \to 0$ (if we let both go to $0$ at the same speed). 
Hence by choosing the parameters $\lambda$ and $\sigma$ small enough we ensure $\E_{\mu} [\eta_{\lambda, \sigma}] <0$.

Now since we assume loss of mass in the k-th orbital $\phi_k$, let us consider $\tilde{\phi}^{(n)}_k (\cdot) = \phi_k (\cdot) + \eta_{\lambda, \sigma}(\cdot - n \Vec{e})$, where $\Vec{e}$ is some unit vector in $\R^3$. 
Denoting the orbitals with $\phi_k$ replaced by $\tilde{\phi}_k$ by $\tilde{\Phi}$, we get
\begin{align*}
    \E_{\mu}[\tilde{\Phi}]
    \leq 
    \E_{\mu}[\Phi]
    +
    T[\eta_{\lambda, \sigma}] + J[|\eta_{\lambda, \sigma}|^2] + E_{xc}[|\eta_{\lambda, \sigma}|^2]
    +
    o(1)
    <
    \E_{\mu}[\Phi]
    +
    o(1).
\end{align*}
So for $n$ large enough we obtain $ \E_{\mu}[\tilde{\Phi}] < \E_{\mu}[\Phi] $ and $||\tilde{\Phi}|| > ||\Phi||$.

We can now repeat these steps until we have constructed a $\Psi $ with $\Psi \in \A$ (after a suitable unitary transformation), and arrive at the contradiction
\begin{align*}
    \E_{\mu}[\Psi]
    < 
    \E_{\mu}[\Phi]
    \overset{\eqref{eq:inf_upper}, \eqref{eq:inf_lower}}{\leq}
    \inf \limits_{\A} \E_{\mu}.
\end{align*}
Hence there is no loss of mass, $\Phi \in \A$, and (ii) holds.

Before we move on to (iii) and (iv) we mention that due to the upper and lower bound above,  $\E_{\mu^{(\nu)}}[\Phi^{(\nu)}] \to \E_{\mu}[\Phi]$. But this gives us $T[\Phi^{(\nu)}] \to T[\Phi]$ since the other terms in the energy functional, namely $V, J_{h}, E_{xc}$, are continuous on $L^{2} \cap L^{4}$, and similarly to the proof of Lemma \ref{lowersemicontinuity} we infer strong convergence in $L^{p}$ for $p \in [2,6)$, since there is no loss of mass.

Hence the kinetic energy converges as well, which means $||\nabla \Phi^{(\nu)}||_{2} \to ||\nabla \Phi||_{2}$, giving us 
$\nabla \Phi^{(\nu)} \to \nabla \Phi$ in $L^2$, so $ \Phi^{(\nu)} \to \Phi$ in $H^1(\R^3)^n$.

\textbf{Ad (iii):}
As in \cite{Friesecke-kniely} the statements (iii) and (iv) are more difficult
since the HOMO and LUMO orbitals $\varphi^{(\nu)}_H$ and $\varphi^{(\nu)}_L$ are not defined via universal but $\varphi^{(\nu)}$-dependent sets
and hence an admissible trial function $\psi$ for the limiting HOMO and LUMO orbitals $\varphi_H$ and $\varphi_L$ may not be admissible for the variational principle for the approximating orbitals.
In short, our argument in \eqref{eq:beginning_argument} is not valid anymore.

Hence we need to look at the $L^2$-projector of the sequence $\Phi^{(\nu)} = (\varphi_{1}^{(\nu)}, \ldots,\varphi_{n}^{(\nu)} )$, i.e. $\gamma_{\Phi^{(\nu)}} \chi := \sum_{k =1}^{n} \langle \varphi_{k}^{(\nu)}, \chi \rangle \varphi_{k}^{(\nu)}$.
For any given $\chi \in L^2(\R^3)$ the mapping $\Phi \mapsto \gamma_{\Phi}\chi$ is strongly continuous from $H^1(\R^3)^n$ to $H^1(\R^3)$, i.e.
\begin{equation}\label{eq:convergence_projector}
    \gamma_{\Phi^{(\nu)}} \chi 
    \to 
    \gamma_{\Phi} \chi \text{ in } H^1(\R^3), 
    \quad
    ||\gamma_{\Phi^{(\nu)}} \chi||_2
    \to 
    ||\gamma_{\Phi} \chi||_2.
\end{equation}{}

Furthermore by definition  we have $\gamma_\Phi \chi = \chi$ and $||\chi||_2 =1$ for any $\chi \in \A^H_\Phi$. 
So by \eqref{eq:convergence_projector} we have $||\gamma_\Phi^{(\nu)} \chi||_2 >0$ for all $\nu$ large enough and therefore by the variational principle for the HOMO \eqref{eq:definition_homo_lumo}
\begin{equation} \label{eq:convergence_excitation_homo}
    \E_{\mu^{(\nu)}, \Phi^{(\nu)}}[\varphi_{H}^{(\nu)}]
    \geq 
    \E_{\mu^{(\nu)}, \Phi^{(\nu)}}\bigg[ \frac{\gamma_\Phi^{(\nu)} \chi}{||\gamma_\Phi^{(\nu)} \chi||_2} \bigg]
    = 
    \frac{1}{||\gamma_\Phi^{(\nu)} \chi||_{2}^{2}}
    \E_{\mu^{(\nu)}, \Phi^{(\nu)}}[\gamma_\Phi^{(\nu)} \chi]
    \to 
    1 \cdot \E_{\mu, \Phi}[\chi].
\end{equation}

Here we used the strong convergence of \eqref{eq:convergence_projector} and the fact that the map $(\Phi, \chi, \mu) \mapsto \E_{\mu, \rho}[\chi]$ 
is weak $\times$ strong $\times$ weak$^*$ continuous on $ H^1(\R^3)^n \times H^1(\R^3) \times \mathcal{M}$ due to Lemma \ref{lem:bounds_excitation_energy}.

Since \eqref{eq:convergence_excitation_homo} holds for any $\chi \in \A^H_\Phi$, we have
\begin{equation}\label{eq:homo_upper}
    \liminf_{\nu \to \infty}
     \E_{\mu^{(\nu)}, \Phi^{(\nu)}}[\varphi_{H}^{(\nu)}]
     \geq 
     \sup_{\chi \in \A^{H}_{\Phi}} \E_{\mu,\Phi}[\chi]. 
\end{equation}

Next we prove that $\varphi_H \in A^{H}_\Phi$. 
Since by definition the $\varphi_{H}^{(\nu)}$ lie in the span of $\Phi^{(\nu)}$, we obtain by the weak convergence of $\varphi_{H}^{(\nu)}$ in $H^1(\R^3)$ and the strong convergence of $\Phi^{(\nu)}$ in $H^1(\R^3)^n$ that 
\begin{align*}
    \varphi_{H}^{(\nu)}
    = 
    \gamma_{\Phi^{(\nu)}} \varphi_{H}^{(\nu)}
    =
    \sum_{k =1}^{n} \langle \varphi_{k}^{(\nu)}, \varphi_{H}^{(\nu)} \rangle \varphi_{k}^{(\nu)}
    ~~
    \longrightarrow 
     ~~
     \sum_{k =1}^{n} \langle \varphi_{k}, \varphi_{H} \rangle \varphi_{k}
     =
     \gamma_{\Phi} \varphi_{H}.
\end{align*}
So taking the limit yields $\varphi_H = \gamma_\Phi \varphi_H$ and hence $\varphi_H^{(\nu)} \to \varphi_H$ strongly, so $\varphi_H \in \A^{H}_\Phi$. 

Furthermore by the continuity properties proven in Lemma \ref{lem:bounds_excitation_energy} we obtain
\begin{equation}\label{eq:homo_lower}
    \lim_{\nu \to \infty}
     \E_{\mu^{(\nu)}, \Phi^{(\nu)}}[\varphi_{H}^{(\nu)}]
     = 
      \E_{\mu,\Phi}[\varphi_{H}]. 
\end{equation}
Combining \eqref{eq:homo_upper} and \eqref{eq:homo_lower} yields (iii).

\textbf{Ad (iv):}

The corresponding proof for the LUMO starts similarly. 
Take any $\chi \in \A^L_\Phi$, i.e. $\gamma_\Phi \chi = 0$ and $||\chi||_2 =1$. 
Then again by \eqref{eq:convergence_projector} and the variational principle for the LUMO \eqref{eq:definition_homo_lumo} we obtain
\begin{equation}
    \E_{\mu^{(\nu)}, \Phi^{(\nu)}}[\varphi_{L}^{(\nu)}]
    \leq 
    \E_{\mu^{(\nu)}, \Phi^{(\nu)}}\bigg[\frac{(I - \gamma_\Phi^{(\nu)}) \chi}{||(I - \gamma_\Phi^{(\nu)} ) \chi||_2}\bigg]
    = 
    \frac{1}{||(I - \gamma_\Phi^{(\nu)} )\chi||_{2}^{2}}
    \E_{\mu^{(\nu)}, \Phi^{(\nu)}}[\gamma_\Phi^{(\nu)} \chi ]
    \to 
    1 \cdot \E_{\mu, \Phi}[\chi].
\end{equation}
Minimizing over $\chi \in A^{H}_L$ gives 

\begin{equation} \label{eq:upper_bound_lumo}
     \limsup_{\nu \to \infty} \E_{\mu^{(\nu)}, \Phi^{(\nu)}}[\varphi_{L}^{(\nu)}]
    \leq
   \inf_{\chi \in \A^L_\Phi} \E_{\mu, \Phi}[\chi].
\end{equation}

For the lower bound we use the weak $\times$ weak $\times$ weak$^*$ lower semicontinuity of the map $(\Phi, \mu, \chi) \mapsto \E_{\mu, \Phi}[\chi]$  by Lemma \ref{lem:bounds_excitation_energy}

\begin{equation}\label{eq:lower_bound_lumo}
    \E_{\mu, \Phi}[\varphi_{L}] 
    \leq 
    \liminf_{\nu \to \infty}  \E_{\mu^{(\nu)}, \Phi^{(\nu)}}[\varphi_{L}^{(\nu)}].
\end{equation}

Unfortunately we are not done since the $\varphi_L^{(\nu)}$ are not known to converge strongly but just weakly, so we still need to prove that $\varphi_L$ is admissible.
The limit $\varphi_L$ lies in the orthogonal complement of the $(\varphi_k)_{k=1}^{n}$, since $\Phi^{(\nu)}$ converges strongly and therefore
\[
0 = \langle \varphi_L^{(\nu)} , \varphi_k^{(\nu)} \rangle
\gegen{\nu}{\infty} 
\langle \varphi_L , \varphi_k \rangle
\qquad \forall k \in \{ 1, \ldots, n \}.
\]
So we only need to prove $||\varphi_L||_2 =1$. 
Assume $||\varphi_L||_2  < 1$. Then by \eqref{eq:upper_bound_lumo} and \eqref{eq:lower_bound_lumo}
\begin{equation}
    I_1 
    =
    \inf_{\chi \in \A^L_\Phi} \E_{\mu, \Phi}[\chi] 
    \geq 
    \E_{\mu, \Phi}[\varphi_{L}]
    =
    \E[\varphi_L]
    \geq 
    I_{||\varphi_L||}.
\end{equation}
But we proved in Lemma \ref{lem:fundamental_properties_infimum} that the mass--to--LUMO--energy map $\lambda \mapsto I_\lambda$ is strictly decreasing. Hence $||\varphi_L||_2 = 1$, so $\varphi_L$ is admissible and together with \eqref{eq:upper_bound_lumo} and \eqref{eq:lower_bound_lumo} this establishes (iv).

Before we prove part b) let us make some remarks.
Since there is no loss of mass, $\varphi_{L}^{(\nu)}$ converges strongly in $L^2(\R^3)$ and hence strongly in $L^p(\R^3)$ for $p \in [2,6)$. But since $ \E_{\mu^{(\nu)}, \Phi^{(\nu)}}[\varphi_{L}^{(\nu)}] \to \E_{\mu, \Phi}[\varphi_L]$ and all terms except the kinetic energy converge due to the continuity results in Lemma \ref{lem:bounds_excitation_energy}, we must also have 
$T[\varphi_{L}^{(\nu)}] \to T[\varphi_L]$, i.e. $\varphi_L^{(\nu)}$ converges strongly in $H^1(\R^3)$.

In conclusion, we know that $\big(\Phi^{(\nu)}, \varphi_{H}^{(\nu)}, \varphi_{L}^{(\nu)} \big)$ converges strongly in $H^1(\R^3)^{(n+2)}$.

\textbf{Part (b)}

To prove sequential compactness of the set $\mathcal{B}$ is now quite easy since by part a) and Banach-Alaoglu we just need to prove that any sequence $\big(\Phi^{(\nu)}, \varphi_{H}^{(\nu)}, \varphi_{L}^{(\nu)}, \mu^{(\nu)} \big) \in \mathcal{B}$ is bounded in $H^1(\R^3)^{(n+2)} \times \M$.
For $\Phi^{(\nu)}$ and $\varphi_{L}^{(\nu)}$ this follows from the bounds in Lemma \ref{lem:bounds_ks} and \ref{lem:bounds_excitation_energy}, respectively, noting that the exponent $p \in [1, \tfrac{5}{3} )$ in the assumption on $v_{xc}$ has the property that the exponents $\tfrac{3(p-1)}{2}$ and  $\tfrac{3(p-1)}{2p}$ of $||\rho||_3$ are strictly less than 1.
Since the $\Phi^{(\nu)}$ stay bounded, so do the $\varphi_{H}^{(\nu)}$. 
Lastly, since $\mu^{(\nu)} \geq 0$ we have $||\mu^{(\nu)}||_{\M} = \intr \di \mu^{(\nu)} = Z$, which concludes the proof.
\end{proof}

As an example of a control goal we consider bandgap tuning as introduced in \cite{Friesecke-kniely}. 
Here the quantity which one wants to influence by a suitable choice of the nuclear charge distribution $\mu$ is the HOMO-LUMO bandgap $\epsilon_H - \epsilon_L$,
where $\epsilon_H$ and $\epsilon_L$ stand for the HOMO and LUMO eigenvalues of the KS Hamiltonian \eqref{eq:ks-canonical}. Any bandgap tuning functional promoting a desired target value $\epsilon_{\ast}$ has to reach its minimum when $\epsilon_H - \epsilon_L = \epsilon_{\ast}$. A simple choice suggested in \cite{Friesecke-kniely} is 
\begin{equation}\label{eq:bandgap}
    J[\Phi, \phi_H,\phi_L, \mu] 
    = \big| \epsilon_L - \epsilon_H - \epsilon_{\ast} \big| ^{2}
    = \big| \E_{\mu, \rho}[\phi_L] - \E_{\mu, \rho}[\phi_H] - \epsilon_{\ast} \big|^2. 
\end{equation}

\begin{theorem}
For any $\epsilon_{\ast}>0$ and for $Z >N$, there exists a nuclear charge distribution $\mu \in \A_{nuc}$ which minimizes the bandgap tuning functional \eqref{eq:bandgap} over $\A_{nuc}$ subject to the constraints \eqref{eq:definition_homo_lumo}.
\end{theorem}

\begin{proof}
The bandgap functional $J$ in \eqref{eq:bandgap} is, due to Lemma \ref{lem:bounds_excitation_energy},  weak $\times$ strong $\times$ strong $\times$ weak$^{*}$ continuous on $(H^1)^n \times H^1 \times H^1 \times \M$.
Hence by the compactness of the set $\mathcal{B}$ proven in Lemma \ref{lem:compactness} it attains its minimum over this set.
\end{proof}{}

\section{Nonexistence of HOMO-LUMO excitations in the neutral case $Z = N$ \label{sec:non_existence}}

We now introduce carefully chosen and realistic model densities $\rho$ and prove that the excitation functional $\E_{\mu, \rho}$ admits no excited states, i.e. no bound states other than the ground state. See the picture on the right in Figure \ref{fig:spectra}.

This finding suggests that also for the true KS ground state density $\rho$, it may happen that there are no exact HOMO-LUMO excitations.
Of course, the Hamiltonian possesses continuous spectrum  above the ground state energy and therefore ``metastable'' excitations (suitable square-integrable superpositions of continuous eigenstates) still exist.

From a mathematical point of view, the results in this section show that the assumption $Z>N$  in our existence result (Theorem \ref{thm:lumo}) was in fact \textit{sharp} and cannot be weakend to $Z\geq N$. 
Note that the model densities considered here satisfy the assumptions of Theorem \ref{thm:lumo}, in particular \eqref{eq:decay}.

\subsection*{H-Atom ground state density }

Let $\mu = \delta_0$ and let $\phi_H (x) = \tfrac{1}{\sqrt{\pi}} e^{- |x|}$ be the hydrogen atom ground state for the Schr\"{o}dinger equation, i.e. the lowest eigenfunction of  $- \tfrac{1}{2} \Delta - \tfrac{1}{|\cdot|}$.
Its density is
\begin{equation}\label{eq:hydrogen_density}
    \rho_H(x) = |\phi_H|^2(x) = \frac{1}{\pi} e^{- 2 |x|}.
\end{equation}
We expect this to be a good approximation for the KS density, hence the KS operator $h_{\rho_H}$ should be a good approximation to the self-consistent hydrogen KS operator.

In this case the Hartree-potential $v_H$ can be explicitly computed. The well known result is
\[
v_H = \intr \frac{1}{|x-y|} \rho(y) \di y  = \frac{1}{|x|} - e^{-2 |x|} \bigg( 1 + \frac{1}{|x|}\bigg).
\]
The first term cancels the external potential $v_{ext}$, hence the excitation functional becomes
\[
\E[\chi]
=
\frac{1}{2} \intr |\nabla \chi|^2 
+ 
\intr |\chi|^2 \bigg( - e^{-2 |x|} \bigg(1 + \frac{1}{|x|}\bigg) +v_{xc} \big( \tfrac{1}{\pi} e^{-2 |x|}\big) \bigg)
\]
with corresponding Hamiltonian
\[
h_{\rho_H}
=
- \frac{1}{2} \Delta + V(x)
= 
-\frac{1}{2} \Delta - e^{-2 |x|} \bigg(1 + \frac{1}{|x|}\bigg)   + v_{xc} \big( \tfrac{1}{\pi} e^{-2 |x|}\big) .
\]

\subsection*{Model ground state density for the He-Atom}

In order to construct a model density for Helium, we make the following ansatz. \\
We take a dilated version of the hydrogen orbital, i.e.
\begin{equation} \label{eq:toy_density}
    \varphi_\alpha (x) = \frac{\alpha^{\nicefrac{3}{2}}}{\sqrt{\pi}} e^{- \alpha |x|}  
    \text{ and }
    \rho_\alpha = 2 |\varphi_\alpha|^2,
\end{equation}
and  -- following Hans Bethe -- determine  
the parameter $\alpha$ by
\[
\alpha = \argmin \limits_{\beta >0} E_\beta, \qquad 
E_\beta = 2 T[\varphi_\beta] + 2 V_{ne}^{He}[\varphi_\beta] + \intr \intr \frac{|\varphi_\beta|^2(x) |\varphi_\beta|^2(y)}{|x-y|} \di x \di y.
\]
The last term describing the electron-electron interaction comes from
\[
V_{ee}[| \Vec{\psi}_1   \Vec{\psi}_2   \rangle]
=
\sum_{i <j} \left( \intr \intr \frac{| \Vec{\psi}_i|^2(x) | \Vec{\psi}_j|^2(y) }{|x-y|} - \frac{ \big( \Vec{\psi}_i \cdot  \Vec{\psi}_j^{\ast} \big)(x)  \big( \Vec{\psi}_i^{\ast} \cdot  \Vec{\psi}_j \big)(y)}{|x-y|} \right), 
\]
 with the spinors
$\Vec{\psi}_1  = \begin{pmatrix} \varphi \\  0 \end{pmatrix} $
and $\Vec{\psi}_2 = \begin{pmatrix} 0 \\ \varphi  \end{pmatrix}.$

The energy is easily computed as 
\[
E_\beta = \beta^2 - 4\beta + \tfrac{5}{8} \beta,
\]
which implies
\[
\alpha= \frac{27}{16} = 1.6875 .
\]

\begin{figure}[h]
    \centering
    \begin{subfigure}[c]{.5\textwidth}
    \includegraphics[width = 0.95 \textwidth]{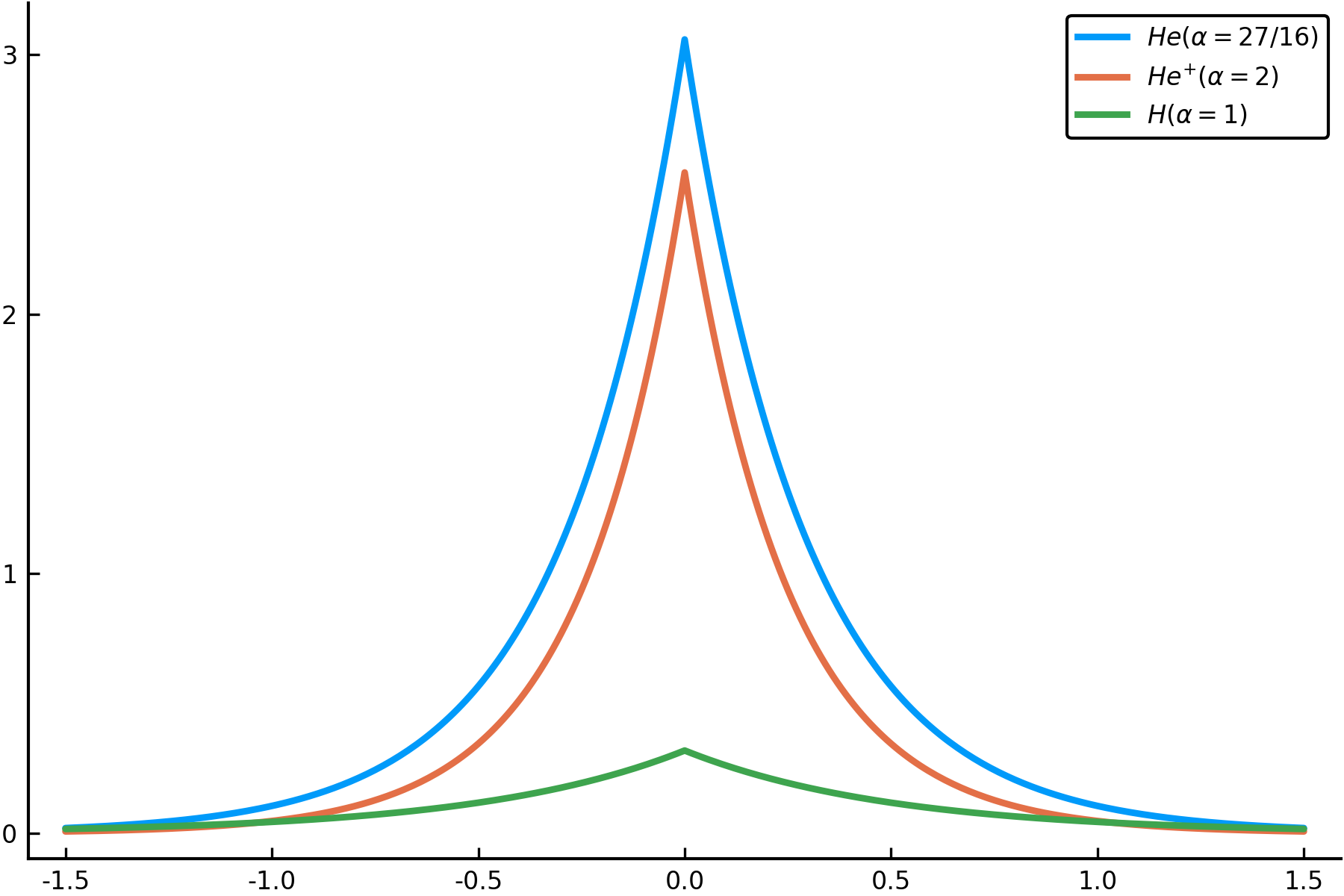}
    \end{subfigure}%
    \begin{subfigure}[c]{.5\textwidth}
    \includegraphics[width = 0.95 \textwidth]{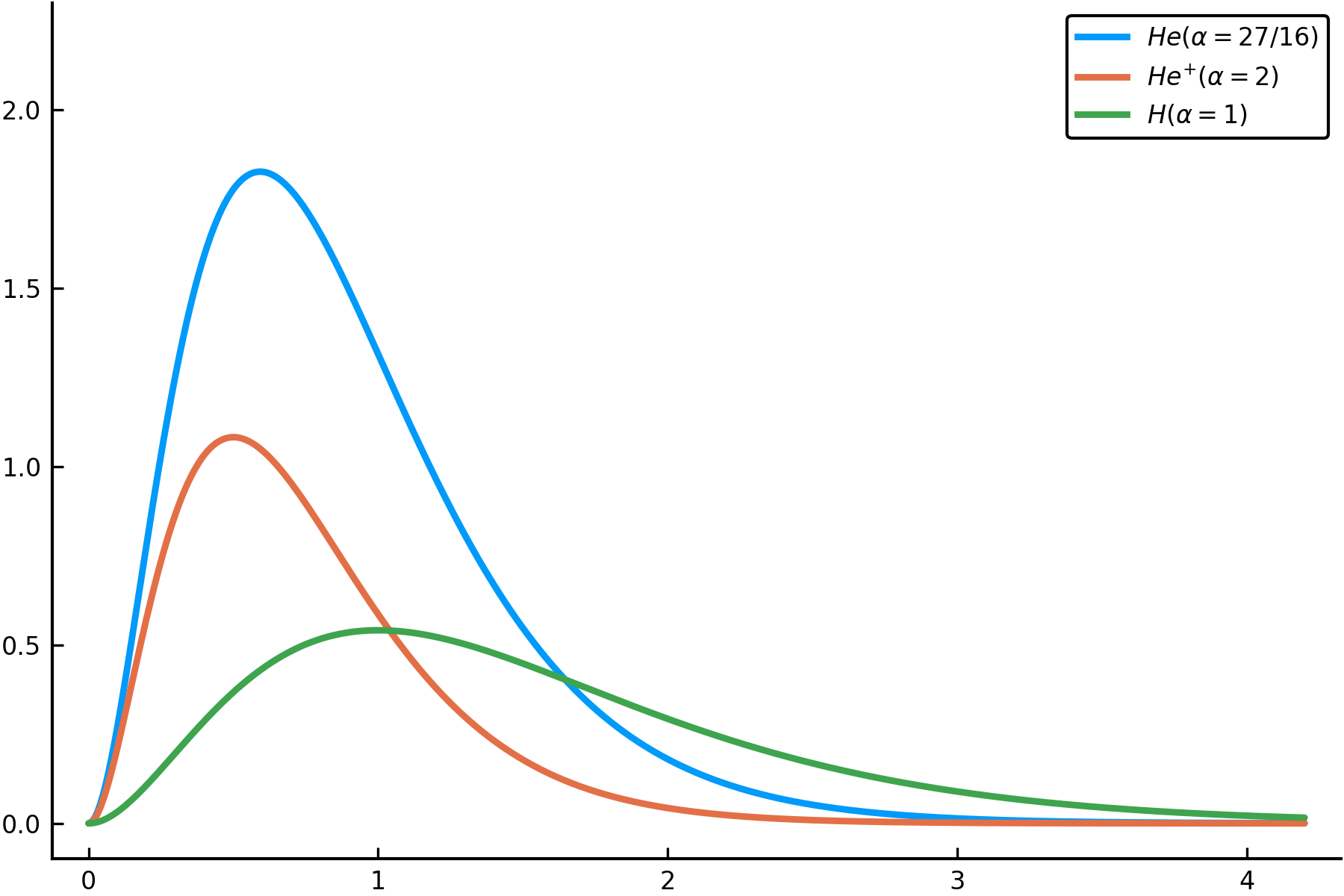}
    \end{subfigure}
    \caption{Comparison of the density (left) and radial density (right) of hydrogen, of the Helium ion He$^+$, and of the model \eqref{eq:toy_density} for Helium. }
    \label{fig:toy_density}
\end{figure}
In Figure \ref{fig:toy_density} one sees that this value corresponds to the fact that the second electron does not see the full Coulomb potential of the nucleus but a screened one.

With this density the Hartree-potential can again be computed explicitly and plugging this density into our excitation functional gives us the Hamiltonian
\[
h_{\rho_{\alpha}} = - \frac{1}{2} \Delta -  2 e^{-2 \alpha |x|} \bigg( \frac{1}{|x|} + \frac{1}{\alpha} \bigg) + v_{xc}(\rho_\alpha(x)).
\]
With these two densities at hand we can state the main result of this section and complete the picture given in Figure \ref{fig:spectra}.

\begin{theorem}[Spectrum of the KS Hamiltonian in the case $Z=N$]\label{thm:nonexistence}
Consider either the hydrogen atom ( $N=1, \mu = \delta_0$) with the density $\rho = \rho_H$ given by \eqref{eq:hydrogen_density} or the helium atom ($ N=2, \mu = 2 \delta_0$) with the density $\rho = \rho_\alpha$ given by \eqref{eq:toy_density}. 
Furthermore let the exchange-correlation energy be given by either Dirac exchange, PW81 or PZ92.
Then the spectrum of the KS Hamiltonian has the form
\begin{equation}
    \sigma \big( h_{\mu, \rho} \big) = \{\epsilon_0\} \cup [0, \infty), \quad \text{for some } \epsilon_0 < 0. 
\end{equation}
In particular, the Hamiltonian possesses exactly one bound state (up to spin in the hydrogen case) and no excited states, i.e. no bound states above the ground state.
\end{theorem}

The result of Theorem \ref{thm:nonexistence} is quite significant from a computational point of view.
In numerical methods one, of course, obtains excited states, but in the limit of complete basis sets in infinite volume these might dissolve into metastable states associated with the continuous spectrum.

\begin{proof}
As in the proof of Theorem \ref{thm:neutral_case} the potential of the KS Hamiltonian is in $L^2 + L^\infty_\epsilon$, hence in both cases $\sigma_{ess} \big( h_{\mu, \rho} \big) = [0, \infty)$. 

Next, we prove that there is at least one bound state with eigenvalue $\epsilon_0 < 0$.
By the Rayleigh–Ritz method it suffices to find a $\psi \in D(h_{\mu, \rho}) = H^1(\R^3)$ with $||\psi||_2 = 1$ and $\langle \psi, h_{\mu, \rho} \psi \rangle <0$.
Since for any LDA functional we have $e_{xc} \leq e_{x}$, it suffices to prove the inequality for Dirac exchange.
As a test function we choose the corresponding orbitals we used in the construction of our densities. 
These terms are easily computed to give for the hydrogen case
\begin{align}
   \epsilon_{0}^{H} \leq  \langle \phi_H, h_{\delta_0, \rho_H} \psi_H \rangle
    &= \frac{1}{2} - \frac{3}{8} - \bigg( \frac{3}{\pi^2}\bigg)^{\nicefrac{1}{3}} \frac{27}{64}   = -0.1587 <0 ,\\
\intertext{and for the helium atom with $\alpha = \tfrac{27}{16}$}
    \epsilon_{0}^{\alpha }\leq \langle \phi_\alpha, h_{2\delta_0, \rho_\alpha} \psi_\alpha \rangle
    & = \frac{\alpha^2}{2} - \frac{2\alpha^2 +1}{4\alpha} - \bigg( \frac{6}{\pi^2} \bigg)^{\nicefrac{1}{3}} \frac{27}{64}\alpha
    = -0.1711<0.
\end{align}
So for both atoms we have at least one bound state.

Now we use the upper bound on the number of bound states given in \cite{glaser1976} which says that the number $N_\ell$ of bound states with angular momentum $\ell$ satisfies
\begin{equation} \label{eq:upper-bound}
   N_\ell \leq  (2l+1)^{1-2p} I_p(V),
\end{equation}
where the functional $I_p(V)$ is given by 
\begin{equation} 
I_p(V) = C_p \intr \frac{1}{4 \pi} |x|^{2p-3} \big( V_{-}(x) \big)^p \di x, \quad \text{ with }C_p = \frac{(p-1)^{p-1} \Gamma(2p)}{p^p \Gamma(p)}. 
\end{equation}

For a radially symmetric potential this reduces to
\begin{equation} \label{eq:upper-bound-integral}
I_p(V) = C_p \int \limits_{0}^{\infty}  \frac{1}{x} \big( x^2 V_{-}(x) \big)^p \di x. 
\end{equation}

We call $I_p$ the Glaser-Martin-Grosse-Thirring (GMGT) functional.
Here $V_{-}$ denotes the negative part of the potential, and the parameter $p \geq 1$ for the radially symmetric case while for the general case we have the restriction $p \geq \tfrac{3}{2}$.

If one calculates this integral numerically for our hydrogen density \eqref{eq:hydrogen_density} with Dirac exchange, one obtains that the minimum value is attained at $p = 1.4$ with $I_p(V) \approx 1.61587$.
This gives us the upper bound
\[
N_\ell < \frac{1}{(2l+1)^{1.8}} 1.61587,
\]
which means $N_0 \leq 1$ and all other $N_\ell = 0$ for $\ell \geq 1$.
So up to the degeneracy with respect to spin there is only the ground state and there are no excited states.

For the helium density \eqref{eq:toy_density} 
we computed the GMGT functional $I_p(V)$ 
for the Dirac \cite{dirac_1930}, the Perdew-Zunger \cite{perdew-zunger} and the Perdew-Wang \cite{perdew-wang} exchange-correlation functional and obtained
\begin{align*}
I_p(V^{D}) &= 1.11465, \text{ at the value } p= 1.65,
\\
I_p(V^{PZ}) &= 1.40558, \text{ at the value } p=  1.5,
\\
I_p(V^{PW}) &= 1.40184, \text{ at the value } p=  1.5.
\end{align*}
\begin{figure}[h]
    \centering
    \includegraphics[width = 0.5 \textwidth]{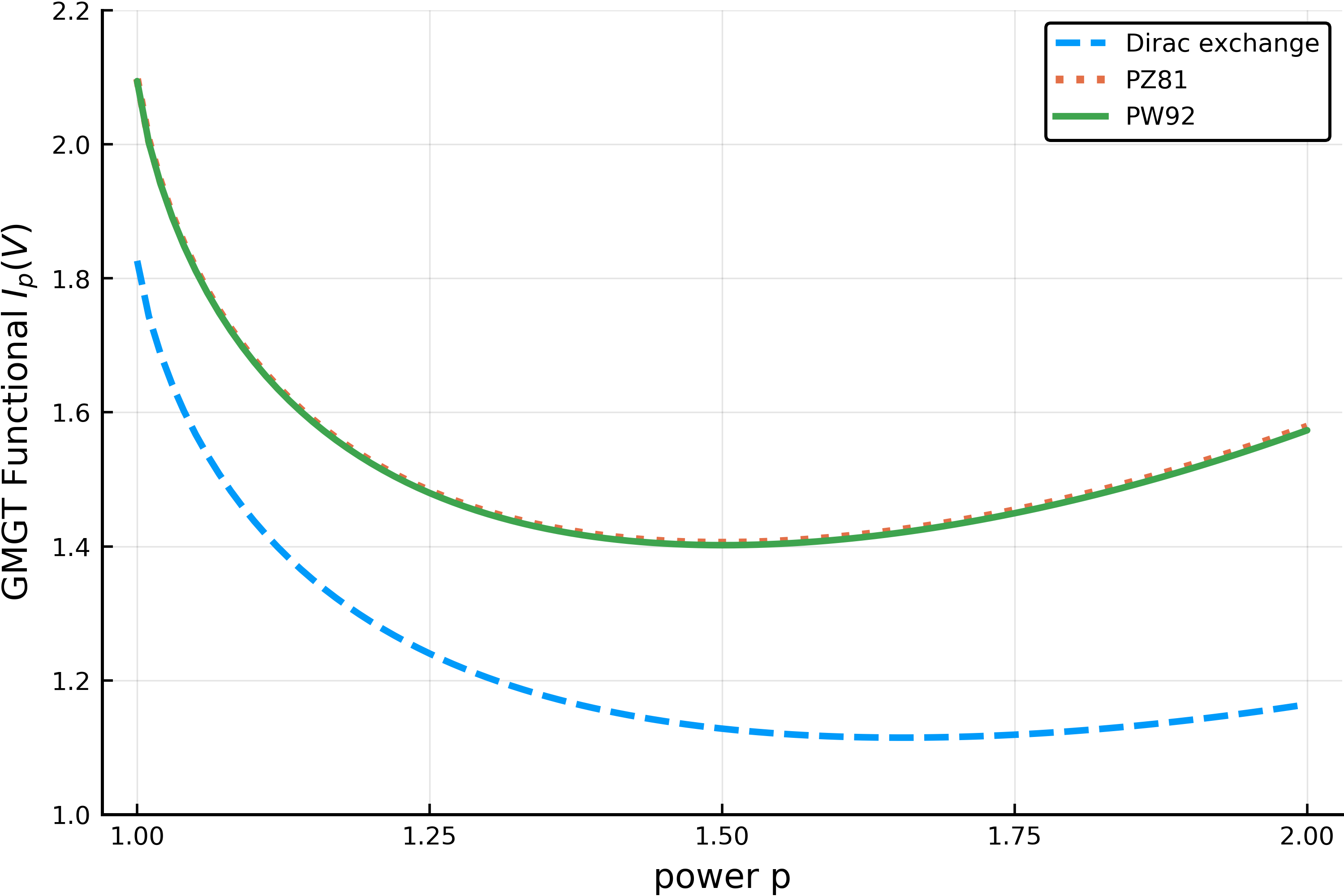}
    \caption{The Glaser-Martin-Grosse-Thirring functional \eqref{eq:upper-bound-integral} of the effective Kohn-Sham potential for different choices of the exchange-correlation functional and the approximate ground state density $\rho_\alpha$ \eqref{eq:toy_density}.
    Values strictly less than 2 mean that there is at most one bouind state.}
    \label{fig:integral_upper_bound}
\end{figure}

Here we denote by $V_{-}$ the potential $V_{-}(x) = 2 e^{-2 \alpha |x|} \big( \frac{1}{|x|} + \frac{1}{\alpha} \big) + v_{xc}(\rho_\alpha(x))$, so in e.g. the Dirac case we have
\[
V_{-}(x) = 2 e^{-2 \alpha |x|} \bigg( \frac{1}{|x|} + \frac{1}{\alpha} \bigg) + \bigg( \frac{3}{\pi} \bigg)^{\tfrac{1}{3}} \frac{2 \alpha}{\pi} e^{-2 \alpha |x|}.
\]
As before, it now follows from \eqref{eq:upper-bound} and \eqref{eq:upper-bound-integral} that no bound state other than the ground state orbital exists. 

Hence for both hydrogen and helium we have exactly one bound state -- the ground state itself -- which corresponds to an eigenfunction with eigenvalue $\epsilon_0 <0$.
\end{proof}

In Figure \ref{fig:integral_upper_bound} the values of $I_p(V)$ as a function of $p$ are given in the case of helium for the three LDAs mentioned above. The figure shows that the minimum value of the GMGT-functional is always attained at some $p\geq 1.5$, i.e. in the interval $p \in [\tfrac{3}{2}, \infty)$, where the upper bound is also valid for non-symmetric potentials.  Hence, as long as the real KS-density is only a small perturbation of our model density $\rho_\alpha$, even if it were not symmetric, our results would hold. 
\\[2mm]
The overall conclusion of this section is that there are no LUMO excitation for the model densities \eqref{eq:hydrogen_density} and \eqref{eq:toy_density}.
\\[2mm]
Our findings raise the following interesting questions which lie beyond the scope of the present paper. First, is the GMGT nonexistence criterion satisfied for numerically obtained ground state densities of hydrogen and helium, or more complex atoms and molecules?
A second question is whether the absence of exact excitations persists for more advanced excitation models like the Casida ansatz.

\vspace*{7mm}

\noindent
\textbf{Acknowledgements.}
This project was supported by the DFG through the IGDK 1754 \textit{Optimization and Numerical Analysis for Partial Differential Equations with Nonsmooth Structures.}
\vspace*{2mm}

\section{Appendix: Analytical properties of PZ81 and PW92}
In the following we verify explicitly that the correlation functionals PZ81 and PW92 satisfy  the assumptions \eqref{eq:assumptions} and \eqref{eq:assumption_at_zero} of this paper.  
In the physics literature the exchange-correlation functional is usually specified by the energy per particle at the density $\rho$, denoted $\epsilon_{c}(\rho)$.
We work with the mathematically convenient energy per unit volume $e_{xc}(\rho) =  \rho \epsilon_{c}(\rho)$.
The exchange part was already discussed in the Remark \ref{rem:assumptions_dirac}, so we only need to deal with the correlation part.

Furthermore, recall the Wigner-Seitz radius
\begin{equation} \label{eq:wigner_seitz}
    r_s(x) := \bigg( \frac{3}{4 \pi \rho(x)} \bigg)^{\tfrac{1}{3}},
\end{equation}
which is a standard parameter in physics to describe the local electron density of a system.
Lastly, we remark that in the following $C$ will describe a generic constant, which may have different values
at each appearance, but is independent of $\rho$ and $r_s$.

\subsection*{Perdew-Wang (PW92)}
In this paper we consider only spin-unpolarized systems, so we have 
$\zeta = \frac{n_{\uparrow} - n_{\downarrow}}{n_{\uparrow} + n_{\downarrow}} =0$. 
So (in the notation of the original paper) we only need to check the assumptions \eqref{eq:assumptions} and \eqref{eq:assumption_at_zero} for $e_{c}(r_{s},0)$.
The PW92 correlation functional is given by
\[
\epsilon_{c} (r_s) = - 2 A (1+\alpha_1 r_s) \log \bigg(  1+ \frac{1}{2A\big( \beta_1 r_{s}^{\nicefrac{1}{2}} + \beta_{2} r_{s} + \beta_{3} r_{s}^{\nicefrac{3}{2}} + \beta_{4} r_{s}^{2}\big)}\bigg).
\]
In order to improve readability of the arguments below, we define the following functions
\[
f(r_s): =  - 2 A (1+\alpha_1 r_s), \qquad g(r_s) := 2A\big( \beta_1 r_{s}^{\nicefrac{1}{2}} + \beta_{2} r_{s} + \beta_{3} r_{s}^{\nicefrac{3}{2}} + \beta_{4} r_{s}^{2} \big).
\]
So returning to the notation of our paper we need to check \eqref{eq:assumptions} and \eqref{eq:assumption_at_zero} for the function
\[
e_{c}(\rho)
=
\rho \cdot \epsilon_{c} (r_s(\rho))
=
\rho  f(r_s(\rho)) \log \bigg( 1 + \frac{1}{g(r_s(\rho))} \bigg).
\]
This function is clearly continuously differentiable for $\rho>0$, so we only check the limit $\rho \to 0$.
Since $f(r_s(\rho)) = O(r_s) = O(\rho^{-\nicefrac{1}{3}})$ and $g(r_s) = \Omega(\sqrt{r_s})$ for $\rho \to 0$, repectively $r_s \to \infty$, we get
\[
\limsup \limits_{\rho \to 0}  | e_c(\rho) |  \leq
\limsup \limits_{\substack{\rho \to 0 \\r_s \to \infty }} C   \rho^{\nicefrac{2}{3}} \log \big( 1  + C r_{s}^{-\nicefrac{1}{2}}\big) =0,
\]
hence with $e_c(0) =0$, $e_{c}$ is continuous.

Next we calculate the derivative

\begin{align} \label{eq:perdew_wang}
    v_{c}(\rho) =
    \drho e_c(\rho)
    &=
    f(r_s) \log \big( 1 + \tfrac{1}{g(r_s)} \big)
    +
    \rho   \big(\drho r_{s} \big)   \log \big( 1 + \tfrac{1}{g(r_s)} \big) \drs f(r_s)
    +
    \rho  f(r_s)   \big( \drho r_{s}   \big) \drs \log \big( 1 + \tfrac{1}{g(r_s)} \big) \nonumber \\
\intertext{using $\drho r_s = - \tfrac{1}{3} \tfrac{r_s}{\rho}$ we obtain}  
    v_c ( \rho)  &=
    \log \big( 1 + \tfrac{1}{g(r_s)} \big) \big[ f(r_s) + \tfrac{2 A \alpha_1}{3} r_s \big]  + \tfrac{1}{3} r_s f(r_s)\frac{\drs g(r_s)}{g(r_s) + g(r_s)^2}. \nonumber \\
\intertext{Since $g(r_s)$ consists only of powers of $r_s$ its derivative is $\drs g(r_s) = \tfrac{1}{r_s} \Theta(g(r_s))$ and we get}
v_c ( \rho)  &=
    \log \big( 1 + \tfrac{1}{g(r_s)} \big) \big[ f(r_s) + \tfrac{2 A \alpha_1}{3} r_s \big]  +  f(r_s) \Theta \big( \tfrac{1}{g(r_s) + 1} \big).
\end{align}
Applying the inequality $\log(1 + x) \leq x$ for $x > -1$, using $\tfrac{1}{g(r_s)} = O(r_s^{-2})$ and taking the limit $\rho \to 0, r_s \to \infty$ gives now

\[
\limsup \limits_{\rho \to 0}  | e_c(\rho) |  \leq
\limsup \limits_{\substack{\rho \to 0 \\r_s \to \infty }}
C \tfrac{1}{r_s^2} (1 + r_s) = 0.
\]
So also $v_{c}$ is continuous and hence it suffices to prove $v_{c} \leq c_{xc} (1 + \rho^{p-1})$ for $\rho \to \infty$, i.e. $r_s \to 0$. Using $f(r_s) \to C$,$g(r_s) \to 0$ and $\tfrac{1}{g(r_s)} = O(r_{s}^{-\nicefrac{1}{2}})$ we obtain
\[
|v_{c}(\rho)| 
\leq
C \bigg( 1 + \log \big( 1 + \tfrac{1}{g(r_s)} \big) \bigg)
\leq 
C \big( 1 +r_{s}^{-\nicefrac{1}{2}}\big)
= 
C \big( 1 + \rho^{\nicefrac{1}{6}}\big),
\]
so  \eqref{eq:assumptions}holds with $p= \tfrac{7}{6}$ in.

For \eqref{eq:assumption_at_zero} we can choose $q = \tfrac{17}{12}$ (any value between $\tfrac{4}{3}$ and $\tfrac{3}{2}$ will do). Writing the condition in terms of $r_s$ then gives us 
\[
\limsup \limits_{r_s \to \infty} \big( r_s^3\big)^{q-1}
 f(r_s) \log \bigg( 1 + \frac{1}{g(r_s)} \bigg) <0.
\]
Realizing that $f(r_s) \sim - r_s $ for $r_s \to \infty$ and that $\tfrac{x}{1+x} \leq \log(1+x) \leq x$ implies for the $\log$-term $\Theta (\tfrac{1}{1+ g(r_s)}) = \log(1 + \tfrac{1}{g(r_s)})$, transforms this into
\[
\limsup \limits_{r_s \to \infty} \big( r_s^3\big)^{q-1}
 (- r_s +1 ) \Theta(\tfrac{1}{r_s^2}) <0,
\]
which is true if $q > \tfrac{4}{3}$.

\subsection*{Perdew-Zunger (PZ81)}
Here we again consider spin-unpolarized systems, i.e. $\zeta = 0$, as in the case of PW92. 
The precise value of the PZ81 constants is important for continuity and continuous differentiability; they were chosen in such a way that $e_{c}$ is continuously differentiable.

The PZ81 correlation functional is given by
\begin{equation}
    \epsilon_c(r_s) : = 
    \begin{cases}
     \frac{\gamma}{ 1 +\beta_1 \sqrt{r_s} + \beta_2 r_s} 
     & \text{for } r_s >1, \\
     A\log(r_s) + B + C r_s \log(r_s) + D r_s 
     & \text{for } r_s  \leq 1. 
    \end{cases}
\end{equation}
Here we used the notation of the original paper and hence $\gamma, B,D$ are negative.
For our analysis we now need to consider $e_{c} = \rho \epsilon(r_s(\rho))$.
Hence, we calculate $v_{c}$ to be

\begin{align*}
v_{c} (\rho)
&=
   \begin{cases}
     \frac{\gamma}{ 1 +\beta_1 \sqrt{r_s} + \beta_2 r_s} 
     +
     \rho \big( \drho r_s\big) \drs \big(  \frac{\gamma}{ 1 +\beta_1 \sqrt{r_s} + \beta_2 r_s} \big)
     & \text{for } r_s >1, \\
     \big( A\log(r_s) + B + C r_s \log(r_s) + D r_s \big)
     +
     \rho \big( \drho r_s\big) \drs \big(A\log(r_s) + b + C r_s \log(r_s) + D r_s  \big)
     & \text{for } r_s  \leq 1. 
    \end{cases} \\
&=
\begin{cases}
     \frac{\gamma}{ 1 +\beta_1 \sqrt{r_s} + \beta_2 r_s} 
     +
     \frac{\gamma}{3}  \frac{\tfrac{\beta_1}{2} \sqrt{r_s} + \beta_2 r_s}{ \big(1 +\beta_1 \sqrt{r_s} + \beta_2 r_s \big)^2} 
     & \text{for } r_s >1, \\
     A\log(r_s) + B + C r_s \log(r_s) + D r_s 
     +
     - \tfrac{1}{3} \big(A + C r_s + C r_s \log(r_s) + D r_s  \big)
     & \text{for } r_s  \leq 1. 
    \end{cases} 
\end{align*}

The continuity for $\rho \to 0$ is now checked easily: 
\[
\lim_{\rho \to 0} e_{c} (\rho)
=
\lim_{r_s \to \infty} \frac{3}{4 \pi}  r_{s}^{-3}  \frac{\gamma}{ 1 +\beta_1 \sqrt{r_s} + \beta_2 r_s} = 0
\]
and

\[
\lim_{\rho \to 0} v_{c} (\rho)
=
\lim_{r_s \to \infty} 
\frac{\gamma}{ 1 +\beta_1 \sqrt{r_s} + \beta_2 r_s} 
     +
      \frac{\tfrac{\gamma}{3} \tfrac{\beta_1}{2} \sqrt{r_s} + \beta_2 r_s}{ \big(1 +\beta_1 \sqrt{r_s} + \beta_2 r_s \big)^2} 
= 0.
\]

The continuity of $e_{c}$ and $v_{c}$ at the value $r_s =1$ follows from the choice of constants in the original paper \cite{perdew-zunger}, since 
\[
\lim_{r_s \to 1+} \epsilon_{c}(\rho(r_s)) = \frac{\gamma}{1 + \beta_1 + \beta_2}
= B + D
=
\lim_{r_s \to 1-} \epsilon_{c}(\rho(r_s))
\]
and
\[
\lim_{r_s \to 1+} v_{c}(\rho(r_s)) 
= \frac{\gamma}{1 + \beta_1 + \beta_2} +  \frac{\gamma}{3}  \frac{\tfrac{\beta_1}{2} + \beta_2 }{ \big(1 +\beta_1  + \beta_2 \big)^2} 
= B + D - \tfrac{1}{3} (D + A + C)
=
\lim_{r_s \to 1-} v_{c}(\rho(r_s)).
\]

Since $v_{c} $ is continuous,  we only need to check the bound \eqref{eq:assumptions} for $\rho \to \infty$, i.e. $r_s \to 0$. But here the only term which is not bounded is $A \log(r_s)$ for which we again use a standard $\log$-bound, $0 \geq \log(x) \geq 1 - \tfrac{1}{x}$ for $ 0 < x \leq 1$, so $ | v_{c} (r_s)| \leq C \big( 1 + \tfrac{1}{r_s} \big)$ for $r_s$ small enough.
In terms of $\rho $ this means $|v_{c}(\rho)| = C \big( 1 + \rho^{\nicefrac{1}{3}} \big)$, so  \eqref{eq:assumptions} holds with $p= \tfrac{4}{3}$. \\
Also, \eqref{eq:assumption_at_zero} holds with $q = \tfrac{4}{3}$, since then plugging in the relation \eqref{eq:wigner_seitz} between $r_s$ and $\rho$ \eqref{eq:wigner_seitz} yields
\[
\limsup_{ \rho \to 0}  \frac{e_c(\rho)}{\rho^q}
=
\limsup_{r_s \to \infty} \frac{\gamma r_s \big( \tfrac{4 \pi}{3}\big)^{\nicefrac{1}{3}} }{ 1 +\beta_1 \sqrt{r_s} + \beta_2 r_s} 
=
\big( \tfrac{4 \pi}{3}\big)^{\nicefrac{1}{3}} \tfrac{\gamma}{\beta_2} <0.
\]

\begin{small}
\begin{spacing}{0.001} 

\addcontentsline{toc}{section}{References}
\bibliographystyle{plain}
\bibliography{references}


\end{spacing}
\end{small}

\end{document}